\documentclass[aps,10pt,twocolumn,showpacs,pra,citeautoscript,amsmath,amssymb,floatfix,nofootinbib,superscriptaddress,longbibliography]{revtex4-2}

\usepackage{ytableau}
\usepackage{array}

\usepackage{amsmath}
\usepackage{amssymb}
\usepackage{bbm}
\usepackage{amsthm, thm-restate}

\usepackage{graphicx}

\usepackage{xcolor}
\usepackage{multirow}

\usepackage{algpseudocode}
\usepackage{enumerate}

\usepackage[normalem]{ulem}

\theoremstyle{plain}

\newcounter{colN}

\newtheorem{theorem}{Theorem}
\newtheorem{lemma}{Lemma}

\newtheorem{definition}{Definition}
\newtheorem{example}{Example}

\newcommand{\ket}[1]{|#1\rangle}
\newcommand{\bra}[1]{\langle #1|}

\newcommand{\ketbra}[2]{|#1\rangle\langle #2|}

\newcommand{\sgn}{\mathrm{sgn}}
\newcommand{\Tr}{\mathrm{Tr}}
\newcommand{\U}{\mathrm{U}}

\newcommand{\cm}{\mathcal M}
\newcommand{\cn}{\mathcal N}
\newcommand{\ch}{\mathcal H}

\newcommand{\nats}{\mathbb{N}}

\newcommand{\comp}{\mathbb{C}}

\newcommand{\Zl}{\mathbb{Z}}

\newcommand{\spatial}{\mathrm{spatial}}
\newcommand{\spin}{\mathrm{spin}}
\def\Rl{{\mathchoice
{\setbox0=\hbox{$\displaystyle\rm R$}\hbox{\hbox to0pt
{\kern0.4\wd0\vrule height0.9\ht0\hss}\box0}}
{\setbox0=\hbox{$\textstyle\rm R$}\hbox{\hbox to0pt
{\kern0.4\wd0\vrule height0.9\ht0\hss}\box0}}
{\setbox0=\hbox{$\scriptstyle\rm R$}\hbox{\hbox to0pt
{\kern0.4\wd0\vrule height0.9\ht0\hss}\box0}}
{\setbox0=\hbox{$\scriptscriptstyle\rm R$}\hbox{\hbox to0pt
{\kern0.4\wd0\vrule height0.9\ht0\hss}\box0}}}}

\def\Cl{{\mathchoice
{\setbox0=\hbox{$\displaystyle\rm C$}\hbox{\hbox to0pt
{\kern0.4\wd0\vrule height0.9\ht0\hss}\box0}}
{\setbox0=\hbox{$\textstyle\rm C$}\hbox{\hbox to0pt
{\kern0.4\wd0\vrule height0.9\ht0\hss}\box0}}
{\setbox0=\hbox{$\scriptstyle\rm C$}\hbox{\hbox to0pt
{\kern0.4\wd0\vrule height0.9\ht0\hss}\box0}}
{\setbox0=\hbox{$\scriptscriptstyle\rm C$}\hbox{\hbox to0pt
{\kern0.4\wd0\vrule height0.9\ht0\hss}\box0}}}}
\def\I{{\mathbb I}}

\makeatletter
\newcommand*{\balancecolsandclearpage}{%
  \close@column@grid
  \clearpage
  \twocolumngrid
}
\makeatother

\makeatletter
\def\tocdepth@fullmunge{%
\let\l@section@saved\l@section
\let\l@section\@gobble@tw@
\let\l@subsection@saved\l@subsection
\let\l@subsection\@gobble@tw@
}%
\def\tocdepth@fullrestore{%
\let\l@section\l@section@saved
\let\l@subsection\l@subsection@saved
}%
\newcommand{\hidetoc}[0]{\addtocontents{toc}{\string\tocdepth@fullmunge}}
\newcommand{\restoretoc}[0]{\addtocontents{toc}{\string\tocdepth@fullrestore}}
\makeatother

\newcommand{\IQOQI}{Institute for Quantum Optics and Quantum Information,\\ Austrian Academy of Sciences, Boltzmanngasse 3, A-1090 Vienna, Austria}
\newcommand{\Peri}{Perimeter Institute for Theoretical Physics, 31 Caroline Street North, Waterloo, ON N2L 2Y5, Canada}
\newcommand{\VCQ}{Vienna Center for Quantum Science and Technology (VCQ), Faculty of Physics,\\ University of Vienna, Boltzmanngasse 5, A-1090 Vienna, Austria}

\newcommand{\nocontentsline}[3]{}
\newcommand{\tocless}[2]{\bgroup\let\addcontentsline=\nocontentsline#1{#2}\egroup}

\makeatletter
\renewcommand*\l@subsection{\@dottedtocline{1}{1.5em}{2em}}
\makeatother

\usepackage{natbib}
\usepackage{hyperref}
\usepackage{url}
\usepackage[capitalise]{cleveref}
\crefname{figure}{Figure}{figures}
\usepackage{xr-hyper}      

\begin{document}

{
\makeatletter
\def\frontmatter@thefootnote{%
  \altaffilletter@sw{\@fnsymbol}{\@fnsymbol}{\csname c@\@mpfn\endcsname}%
}%
\makeatother

\title{Invariance under quantum permutations rules out parastatistics}

\author{Manuel Mekonnen}
\email{manuel.mekonnen@oeaw.ac.at}
\affiliation{\IQOQI{}}
\affiliation{\VCQ{}}
\author{Thomas D.\ Galley}
\affiliation{\IQOQI{}}
\affiliation{\VCQ{}}
\author{Markus P.\ M\"uller}
\affiliation{\IQOQI{}}
\affiliation{\VCQ{}}
\affiliation{\Peri{}}

\date{April 30, 2026}

\begin{abstract}
Quantum systems invariant under particle exchange are either Bosons or Fermions, even though quantum theory in principle admits more general behavior under permutations. But why do we not observe such \textit{paraparticles} in nature? The analysis of this question was previously limited primarily to specific quantum field theory models. Here we give two distinct model-independent arguments that rule out parastatistics, i.e.\ fundamentally indistinguishable quantum systems transforming under higher-dimensional representations of the symmetric group, which draw on quantum information theory and recent research on internal quantum reference frames. First, we introduce a notion of \textit{complete invariance}: quantum systems should not only preserve their local state under permutations, but also the quantum information they carry about other systems, in analogy to the notion of complete positivity in quantum information theory. Second, we demand that quantum systems are invariant under \textit{quantum permutations}, i.e.\ permutations conditioned on values of permutation-invariant observables. For both, we show that the respective principle is fulfilled if and only if the particle is a Boson or Fermion. Our results show how  quantum reference frames can shed light on a longstanding problem of quantum physics, they underline the crucial role played by the compositional structure of quantum information, and demonstrate the explanatory power but also subtle limitations of recently proposed quantum covariance principles.
\end{abstract}

\maketitle

\section*{Introduction}

It is a fundamental empirical finding that all known particles in our universe are either Bosons or Fermions. But why is this the case? A standard but incomplete argument sometimes presented in the textbooks is as follows. Relabelling of identical particles should preserve all physical predictions. Thus, given a wavefunction $\psi(x_1,\ldots,x_N)$ for $N$ particles, exchanging, say, the first two particles should only result in a global phase,
\begin{align}
    \psi(x_2,x_1,x_3,\ldots,x_N)=e^{i\theta}\psi(x_1,x_2,\ldots,x_N),
    \label{main:eq:symmetrPost}
\end{align}
which implies that the wavefunction transforms under a one-dimensional representation of the permutation group. It is either symmetric, $e^{i\theta}=+1$ (as for Bosons), or antisymmetric, $e^{i\theta}=-1$ (as for Fermions). This requirement is sometimes known as the symmetrization postulate~\cite{Messiah_Greenberg}.

However, it has been understood since the 1950s~\cite{Green} that quantum theory admits more general behaviors of quantum systems under particle exchange. After all, general quantum states are represented by density operators, not state vectors, and these can preserve their form under actions of the permutation group which are more general than multiplication by a complex phase. 
This would lead to neither Bosonic nor Fermionic statistics, but \textit{parastatistics}. Indeed, hypothetical \textit{paraparticles} have been extensively studied
over the last few decades~\cite{Green,Taylor1,Taylor2,Taylor3,Bialynicki-Birula,Greenberg,Ohnuki,LandshoffStapp,Goyal,Sanchez}.

\begin{figure}[hbt]
\centering 
\includegraphics[width=\columnwidth]{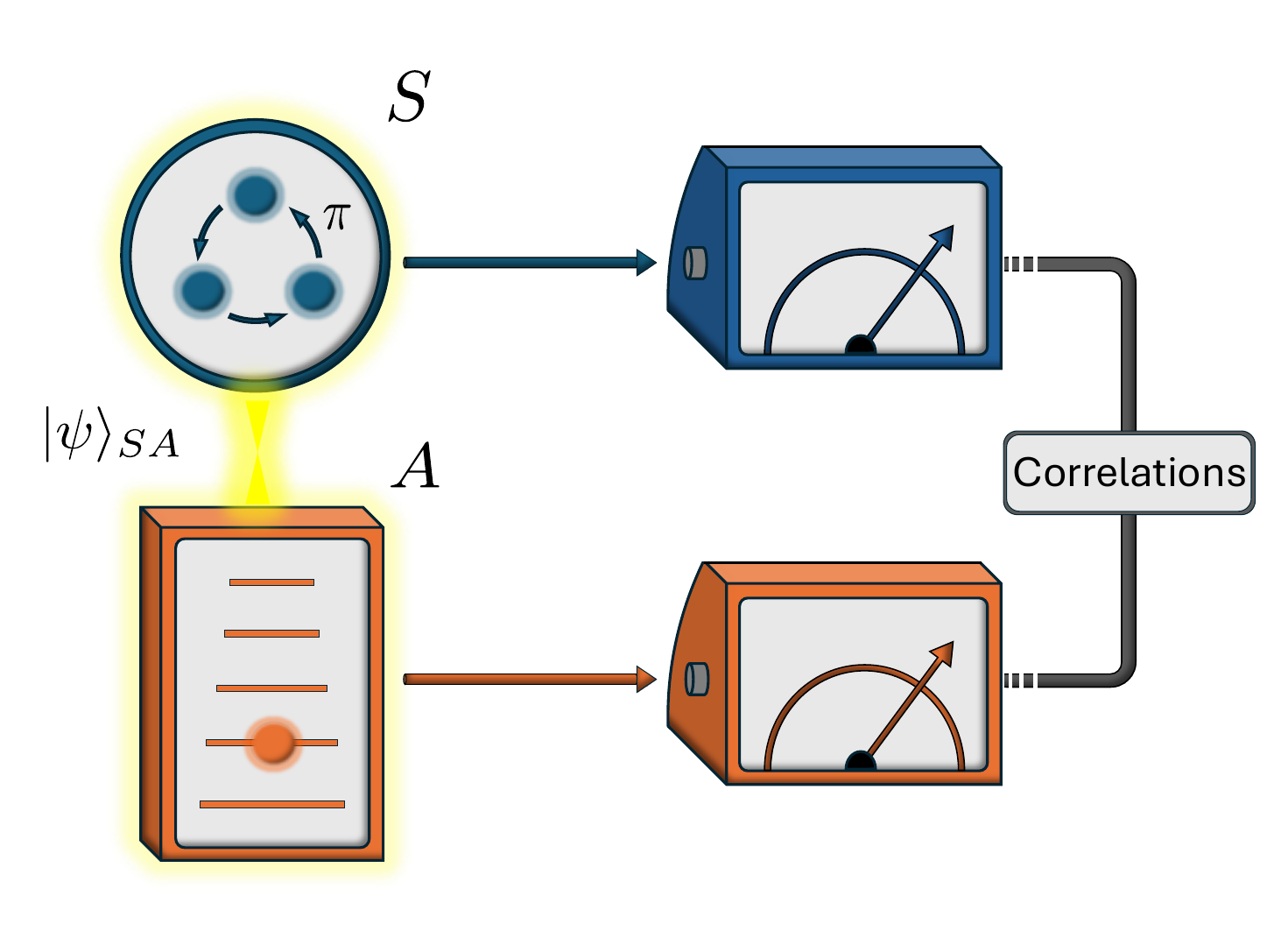}
\caption{\textbf{Complete invariance.} All physical predictions are invariant under permutations $\pi$ of the system of particles $S$. This includes the statistics and correlations of measurements of $S$ and any ancillary system $A$, even if $SA$ is potentially in a pure entangled state $\ket{\psi}_{SA}$.}
\label{main:fig:completeinvariance}
\end{figure}

A widely accepted consequence of this research has been termed the \textit{equivalence thesis}: every consistent theory of paraparticles is physically equivalent to some theory of regular Bosons or Fermions~\cite{Baker}. However, this conclusion has only been obtained under significant assumptions~\cite{Toppan}, such as a ``locality-inspired charge recombination principle'' formulated in the framework of quantum field theory (QFT). Indeed, recent works have theoretically reinforced the consistent possibility of parastatistics ~\cite{ToppanFirst,Toppan2021Parabosons}, including in certain quantum many-body systems~\cite{WangHazzard,Wang2025}, which evade the usual QFT formalism and its locality principles. It is thus natural to ask: is parastatistics excluded by natural arguments that do not rely on (but still apply in a special case to) the framework of QFT?

\begin{figure}[hbt]
\centering 
\includegraphics[width=\columnwidth]{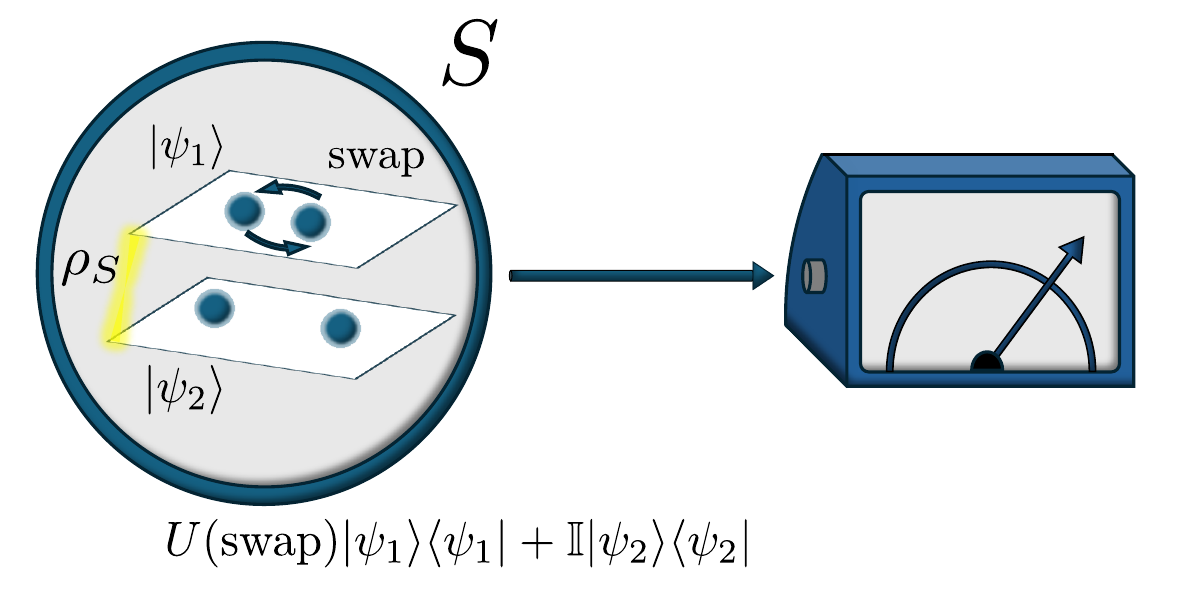}
\caption{\textbf{Invariance under quantum permutations.} The system of particles $S$ with state $\rho_S$ in some superposition of $|\psi_1\rangle$ and $|\psi_2\rangle$ is permuted ``branchwise'', i.e.\ conditioned on the value of a permutation-invariant observable. For example, we could apply $U(\rm swap)$ and swap the particles or apply the identity $\mathbb{I}$ and do nothing, depending on their distance. This must be possible in a way that preserves the statistics of all  measurements on the system. We also consider what kinds of correlations with ancillary systems $A$ are allowed by global invariance of $SA$.}
\label{main:fig:qrf-invariance}
\end{figure}

We answer this question in the affirmative by reconsidering the basic question of what it even means that a quantum system is permutation-invariant. The standard notion says that all observables of the multi-particle system should remain invariant if the particles are exchanged according to a fixed permutation $\pi$. Implementing this unitarily, $\pi\mapsto U(\pi)$, leads to Bosonic and Fermionic, but also parastatistics. This term has been used to refer to other important families of particles with exotic exchange statistics such as anyons, which transform under the braid group, \textit{emergent} parastatistics from ordinary Bosons and Fermions, and paraparticles where the invariant permutation action does more than just swapping modes, as in graded systems described by Lie superalgebras~\cite{ToppanFirst}.

In this work, we rule out parastatistics arising from fundamentally indistinguishable quantum systems transforming under multidimensional irreducible representations of the symmetric group. While not applying to parastatistics in emergent, anyonic and graded systems (as we will discuss), our results are rather general because they apply to all kinds of quantum systems that can be permuted, not only to particles with spatiotemporal locations. To this end, our approach defines permutation invariance directly at the level of Hilbert space operators, and not on the level of classical configuration space~\cite{Landsman}. We take guidance from quantum information theory and research on quantum reference frames, suggesting two important aspects of permutation-invariance that have, to the best of our knowledge, not been considered in this context.

The first aspect is that a quantum system $S$ (say, of $N$ indistinguishable particles) is more than its local, permutation-invariant state: it may contain \textit{quantum information about other systems} $A$ in the form of entanglement, which should also remain invariant under permutations on $S$ (see Figure~\ref{main:fig:completeinvariance}) -- 
even if $SA$ is potentially in a pure entangled state. As we show, if we assume that any total state on $SA$ is invariant under permuting $S$ locally, then this more general form of invariance, which we call \textit{complete invariance}, singles out Bosonic and Fermionic statistics.

A second aspect is that all physical predictions should remain invariant even if the system of particles $S$ is permuted branch-wise, i.e.\ if the permutation $\pi$ is chosen conditionally on the value of some permutation-invariant observable. For example, we could decide to swap two particles, or to do nothing, depending on the value of their distance $d$, see Figure~\ref{main:fig:qrf-invariance}. While classically equivalent to standard permutation-invariance, we show that invariance under such \textit{quantum permutations} is respected by Bosonic and Fermionic, but not by parastatistics.

Furthermore, since we are able to understand quantum permutations as quantum reference frame (QRF) transformations, our results relate in a surprising way to the field of QRFs. There, the idea is to extend the notion of reference frames, and of transformations between them, to quantum theory~\cite{AharonovSusskindCSR,WWW,AharonovKaufherr,Rovelli1991,Bartlett2007,GourSpekkens,Angelo,Palmer2014,Loveridge2017,Loveridge2018,Giacomini,HametteGalley,GiacominiBrukner2020,Trinity,KrummHoehnMueller,CastroRuiz2021,Hamette2021,Ahmad2022,HoehnKrummMueller,GiacominiBrukner2022,Ludescher2022,QRFIndefiniteMetric,GlowackiThesis,Carette2023,Kabel2024}. In this picture, invariance under standard permutations reflects the lack of a fundamental reference frame for labeling particles, while quantum permutation-invariance expresses the absence of a fundamental quantum reference frame for labeling.

\section*{Results}
\subsection*{Standard permutation invariance and parastatistics}
\refstepcounter{section}\label{main:sec:StandPermInvAndParastat}

The notion of indistinguishability of particles in quantum physics describes the fact that particles with identical properties cannot be identified uniquely \textit{in principle}, contrary to classical intuition. This means that shuffling $N$ particles, either passively by changing the labels in their description, or actively by exchanging them via prescribed physical operations, must preserve all physical predictions.

Capturing this in mathematical terms, we associate to each way of shuffling $N$ systems an element $\pi$ in the permutation group $\mathcal{S}_N$. On a given Hilbert space $\mathcal{H}$, we then have some unitary representation $\pi \mapsto U(\pi)$ of this group, depending on the formalism we use to describe indistinguishable particles. For example, in first quantization, a single particle is described by a Hilbert space $\mathcal{H}_1$ such that the total Hilbert space is given by $\mathcal{H}=\mathcal{H}_1^{\otimes N}$. An element $U(\pi)$ then acts by permuting these $N$ tensor factors. For instance, in the case of $N=2$, the swap operation $U(\rm swap)$ would map a state $|\psi \rangle \otimes |\phi \rangle$ to $|\phi \rangle \otimes |\psi \rangle$. In second quantization, and in particular in the formalism of Green's paraparticles that we include in our analysis, this tensor product structure is lost, and so are the aforementioned permutations of tensor factors. Indeed, attempts to define analogous transformations lead to non-unitary and non-physical maps (see Subsection \hyperref[supp:sec:StandPermSecQuant]{F} in Supplementary Note 4 ). Instead, individual systems to be shuffled are associated to modes $k_i$, representing a property such as position or momentum, which can be occupied by some number $n_{k_i}$ of indistinguishable particles. Bosons (Fermions) are described by a Fock space $\mathcal{F}$ generated by creation operators $a_{k_i}^\dagger$ that fulfill the usual (anti)commutation relations, or more general trilinear relations in the case of paraparticles~\cite{Messiah_Greenberg,Green}, see Subsection \hyperref[supp:sec:ParaStatSecQuant]{B} in Supplementary Note 2 . The $N$-particle subspace $\mathcal{F}_N$ of the Fock space is spanned by vectors
\begin{align}
\label{main:eq:FockStates}
    |k_1,...,k_N\rangle_\mathcal{F} := a_{k_1}^\dagger...a_{k_N}^\dagger|0\rangle_\mathcal{F},
\end{align}
where $|0\rangle_\mathcal{F}$ is the vacuum state. Here, an element $U(\pi)$ represents a physical transport operation~\cite{Roos} that permutes the modes, while it acts mathematically by permuting the labels of the modes, and thus the associated creation operators, via $ U(\pi)a_{k_i}^\dagger U(\pi)^ \dagger = a_{k_{\pi(i)}}^\dagger$.

Regardless of the formalism, the concept of indistinguishability can then be formalized by requiring permutation-invariance for all physical states: A state $\rho$ of indistinguishable particles is physically allowed if
\begin{align}
\label{main:eq:StandPermInv}
U(\pi)\rho U(\pi)^\dagger=\rho    
\end{align}
holds for any permutation $\pi$. In second quantization we will only consider situations where we have $N$ distinct modes with each mode occupied by one particle. On the one hand, this is a necessary restriction to make: if the modes have different occupation numbers then the physical systems they represent are simply distinguishable. On the other hand, this restriction still allows for the possibility of physically interesting parastatistical behavior, such as relative $(2\pi/3)$-phase shifts under coherently controlled permutations (see Example \ref{supp:ex:detectPara} in Supplementary Note 2) that are impossible for Bosons and Fermions~\cite{Roos}.

As a direct consequence of Schur's lemma (see the \hyperref[methods:sec:RepTheoryPermInv]{Methods} for details), permutation-invariant states decompose into a direct sum 
\begin{align}
\rho=\bigoplus_{\lambda}\sigma_{\lambda},
\label{eqMTSuperselection}
\end{align}
where the (subnormalized) density matrices $\sigma_\lambda$ are each contained in invariant subspaces $\mathcal{H}_\lambda$, known as superselection sectors. The different sectors, labeled by the distinct irreducible representations (irreps) $\lambda$ of $U(\pi)$, correspond to different particle types. Bosonic (Fermionic) states are contained in the symmetric (antisymmetric) subspace of $\mathcal{H}$ and fulfill $U(\pi)|b\rangle = |b\rangle$ ($U(\pi)|f\rangle = \sgn(\pi)|f\rangle$), thus transforming according to the trivial (sign) irrep of $\mathcal{S}_N$. All other irreps of the permutation group denote exchange statistics which is neither Bosonic nor Fermionic, known as \textit{parastatistics} and obeyed e.g.\ by \textit{paraparticles}. For examples of Bosonic, Fermionic and parastatistical states in first and second quantization formalisms, see Supplementary Note \hyperref[supp:sec:DescPara]{2}.

\subsection*{Complete Permutation Invariance}
\refstepcounter{section}\label{main:sec:CompleteInvariance}

So far, we have been considering any 
system $S$ of $N$ indistinguishable particles in isolation: implementing invariance on the level of states, i.e.\ demanding $U(\pi)_S\rho_S U(\pi)_S^\dagger=\rho_S$, led to the form of $\rho=\rho_S$ above, which includes the possibility of para\-statistics. However, $S$ will typically be part of a larger quantum system $SA$, consisting of itself and some ancillary system $A$, which may be a different type of system on which the permutations do not act. Then, the quantum state $\rho_S$ will be the local reduced state of an \textit{extension} $\rho_{SA}$, i.e.\ of a quantum state that satisfies $\rho_S={\rm Tr}_A\rho_{SA}$. The assumption that this global state is invariant under local permutations of $S$ means
\begin{align}
\label{eqMTInvLocalPerm}
\big (U(\pi)_S\otimes\I_A \big ) \rho_{SA} \big (U(\pi)_S^\dagger\otimes\I_A\big ) = \rho_{SA}.
\end{align}
In analogy with the notion of complete positivity from quantum information theory, we will say that $\rho_S$ is \textit{completely invariant} if Eq.~(\ref{eqMTInvLocalPerm}) holds for \textit{all} possible extensions $\rho_{SA}$. Complete invariance implies invariance, but the converse is not in general true.

It turns out that complete invariance follows already from the existence of a single pure state $\rho_{SA}=|\psi\rangle\langle\psi|_{SA}$, i.e.\ a \textit{purification}~\cite{NielsenChuang} of $\rho_S$, that satisfies Eq.~(\ref{eqMTInvLocalPerm}). In this case, Eq.~(\ref{eqMTInvLocalPerm}) can be interpreted as saying that not only is the local state of $S$ preserved under permutations (as in standard invariance), but so is the quantum information that it carries about other systems $A$.

It turns out to be the case that complete invariance is a crucial difference between ordinary statistics and parastatistics. Requiring this property rules out the latter:
\begin{theorem}
\label{main:th:CompleteInvariance}
A quantum state $\rho_S$ is completely invariant under permutations if and only if it is fully supported on either the Bosonic or the Fermionic subspace.
\end{theorem}
The full proof of this is deferred to the \hyperref[methods:sec:FormalismCI]{Methods} section, but in essence, it works by virtue of showing that complete invariance of a state $\rho_S$ is equivalent to the existence of a single purification $|\psi\rangle_{SA}$ of $\rho_S$ which is invariant under local permutations $U(\pi)_S \otimes \mathbb{I}_A$ on $S$. Furthermore, given any invariant purification, we find that its marginal on $S$ must be either a Bosonic or Fermionic state, thus ruling out invariant purifications for paraparticles.

Besides being a natural notion from the perspective of quantum information theory, another possible motivation for complete invariance comes from the idea that the universe should always be in a pure quantum state, and that all mixed states of subsystems $S$ should be understood as improper mixtures, arising from entanglement with the rest of the world $A$. This is in line with some versions of the many-worlds interpretation of quantum mechanics, and it resembles an attitude of quantum information theory that has been termed the ``church of the larger Hilbert space''. Furthermore, complete invariance expresses a strong compatibility of symmetries and compositions of systems: if we have a system $S$ in our lab, and we know that it is prepared in a completely invariant state (e.g.\ by measuring the projector onto the Bosonic or Fermionic subspace), we can be sure that \textit{nothing else in the universe will change under local permutations}, regardless of how $S$ is correlated or entangled with other systems $A$. Similarly as for complete positivity, a more detailed analysis of this notion could benefit from the category-theoretic methods of process theories, analyzing the compositional structure of quantum theory in a more systematic way.

While complete invariance is a natural requirement for fundamentally indistinguishable particles, it is not expected to hold in systems with emergent symmetries. This is illustrated by Example \ref{ExEmergent} in Supplementary Note 2 and Subsection \hyperref[SubsecFailure]{D} in Supplementary Note 3. Beyond permutations, it is natural to expect that complete invariance also holds for fundamental gauge symmetries in quantum field theory. In Subsection \hyperref[SuppSecCompleteInv]{E} in Supplementary Note 3, we prove a generalization of Theorem~\ref{main:th:CompleteInvariance}: complete invariance under a compact group $\mathcal{G}$ is equivalent to the system $S$ transforming under a one-dimensional representation of $\mathcal{G}$. Even though it is beyond our mathematical framework, note that Dirac quantization implements gauge symmetries via constraints such as $C|\psi\rangle=0$ (for example, the Gupta-Bleuler condition $C=\partial^\mu A_\mu^+$ in quantum electrodynamics~\cite{Tong}). The gauge transformations $G$ generated by $C$ hence satisfy $G|\psi\rangle=|\psi\rangle$ (or, in more general Yang-Mills theories, $G|\psi\rangle=e^{i\theta}|\psi\rangle$~\cite{Jackiw}), corresponding to one-dimensional representations, even if the gauge group is non-Abelian. Hence, if the mathematical framework can be suitably generalized, we expect complete invariance to hold here, confirming the natural intuition that it expresses ``nothing in the world to change whatsoever'' under gauge transformations.

\subsection*{Quantum Permutation Invariance}
\refstepcounter{section}\label{main:sec:QuantPermAsQRFs}

We now consider another extension of standard permutation invariance which is naturally motivated by quantum reference frames, as we explain further below: \textit{invariance under a quantum permutation group}. In contrast to the previous argument, this one rules out parastatistics in both first and second quantization.

A quantum permutation can be thought of as a collection of permutations, each one acting independently on a different branch of a superposition -- a coherently controlled, conditional permutation. For example, if we have $N=2$ particles, we might swap them if their distance is less than $d$, or otherwise do nothing (see Figure \ref{main:fig:qrf-invariance}), as implemented by
\begin{align}
\label{main:eq:ExampleDistanceSwap}
    V_+=U({\rm swap})\cdot P +\I\cdot(\I-P),
\end{align}
where $U({\rm swap})|x_1,x_2\rangle=|x_2,x_1\rangle$, and $P|x_1,x_2\rangle=|x_1,x_2\rangle$ if $|x_1-x_2|<d$ and $0$ otherwise. The two branches here are defined by the distance between a pair of particles. Generally, the branches can be defined by any projective measurement that is itself permutation-invariant. The projective action on the separate branches must be the same as that of the base representation $U(\pi)$, and if we do the same permutation on each branch, $\vec{\pi}=(\pi,\ldots,\pi)$, then we must obtain the projective action of $U(\pi)$. This leads to the following definition, with an arbitrary phase $\theta(\pi)$ depending on the permutation $\pi$:

\begin{definition}
\label{main:def:QuantumPerm}
A quantum permutation is a unitary transformation of the form
\begin{align}
\label{main:eq:QuantumPermDef}
V(\vec{\pi})=\sum_j e^{i\theta(\pi_j)}U(\pi_j)P_j,
\end{align}
where $\{P_j\}$ is a permutation-invariant projective measurement, i.e.\ $[U(\pi),P_j]=0$ for all $\pi\in S_N$ and all $j$. If $\vec{\pi}\mapsto V(\vec{\pi})$ is a projective representation of $S_N\times S_N\times\ldots\times S_N$, then we call the set $\mathcal{V}$ of all $V(\vec{\pi})$ a \textbf{quantum permutation group}. We call $\mathcal{V}$ maximal if the $P_j$ have finite and smallest possible rank under these requirements.
\end{definition}

Now, while all Bosonic states are invariant under $V_+$, a quantum permutation such as 
\begin{align}
\label{main:eq:ExampleDistanceSwapMinus}
    V_-=-U(\rm swap)\cdot P +\I\cdot(\I-P)
\end{align}
leaves all Fermionic states invariant. Moreover, we find that elements such as $V_+$ and $V_-$ are separately contained in different quantum permutation groups, each of which still leaves all Bosonic \textit{or} all Fermionic states invariant. As it turns out, there exists no such group for all paraparticle states, leading to our main result ruling them out: 

\begin{theorem}
\label{main:th:QuantPermInvRulesOutPara}
For any system $S$ of indistinguishable particles and any admissible $\{P_j\}$, there exist exactly two quantum permutation groups $\mathcal{{V}_+}$ and $\mathcal{V}_-$, respectively containing elements of the following forms:
\begin{itemize}
    \item $V_+(\vec{\pi})= \sum_j U(\pi_j)P_j$,
    \item $V_-(\vec{\pi})= \sum_j \sgn(\pi_j)U(\pi_j) P_j$.
\end{itemize} 
All Bosonic (Fermionic) states are invariant under $\mathcal{V}_+$ ($\mathcal{V}_-$). In contrast, any paraparticle state invariant under either group must be diagonal in the basis $\{P_j\}$. Hence, paraparticle systems invariant under any maximal quantum permutation group can neither be entangled nor interact unitarily with any ancillary system, in contrast to Bosonic and Fermionic systems.
\end{theorem}

Furthermore, if we require quantum permutation-invariance with respect to all possible projective measurements $\{P_j\}$ (which holds for Bosons and Fermions), then every paraparticle system must be completely uncorrelated with its environment. The proofs can be found in the \hyperref[methods:sec:FormalismQuantPerm]{Methods} section. As a direct consequence, one finds that in the case of only considering classical states of indistinguishable particles, i.e. diagonal states in a given basis, requiring this conditional kind of permutation-invariance does not imply more restrictions than standard permutation-invariance. The two notions are thus classically equivalent. Tables~\ref{methods:tab:summarytab1} and~\ref{methods:tab:summarytab2} show the forms of invariant states under standard, quantum and strong quantum permutation-invariance, and the possible correlations with ancillary systems.

For the remainder of this section, we will relate our results to the field of quantum reference frames (QRFs). Initiated decades ago~\cite{AharonovSusskindCSR,WWW,AharonovKaufherr,Rovelli1991} it has recently received a significant amount of attention in quantum information, quantum foundations, and quantum gravity research~\cite{Bartlett2007,GourSpekkens,Angelo,Palmer2014,Loveridge2017,Loveridge2018,Giacomini,HametteGalley,GiacominiBrukner2020,Trinity,KrummHoehnMueller,CastroRuiz2021,Hamette2021,Ahmad2022,HoehnKrummMueller,GiacominiBrukner2022,Ludescher2022,QRFIndefiniteMetric,GlowackiThesis,Carette2023,Kabel2024}. We will now argue that quantum permutations can be understood as QRF transformations. Since indistinguishable particles lack individual identity, there is no natural, physically preferred way to label them. Hence there exists a relabeling group of elements $U(\pi)$ preserving all physical predictions, and we can understand standard permutation invariance as the consequence of a \textit{lack of a fundamental reference frame for labeling}. A choice of labeling then means breaking permutation-invariance and  choosing a reference frame from which to describe the system: for example, we can identify the electron that is closer to Alice's (Bob's) laboratory as the ``first (second) particle''. 

Now, as a natural requirement strengthening this, we argue that there is also no physically preferred way to label particles -- or to identify labeling conventions -- \textit{across branches} of superpositions. For example, what is labeled as ``first particle" on one branch may be called ``second particle" on another. Hence there exists a \textit{quantum} relabeling group preserving all physical predictions, and choices of labeling on each branch constitute a choice of a \textit{quantum reference frame}.

QRF transformations in general are understood as coherently-controlled classical reference frame changes~\cite{HametteGalley}, where the controlling branches are themselves defined in a frame-independent way~\cite{Cepollaro}, or as state-dependent gauge transformations~\cite{Butterfield}. Hence, quantum permutations can be interpreted as QRF transformations. As abstract groups, they are $S_N\times S_N\times\ldots\times S_N$ (one independent choice of permutation per branch), but the specific way they act on a Hilbert space depends on the type of physical system (via $\mathcal{V}_+$ for Bosons, via $\mathcal{V}_-$ for Fermions). This is a well-known general phenomenon in quantum physics: for example, electrons carry a different representation of the rotation group than photons. Indeed, it is well-known that applying an inappropriate quantum permutation group, such as $\mathcal{V}_+$ via active permutations to Fermions, will violate invariance and lead to physically detectable relative phases~\cite{Roos}. Hence, elements of $\mathcal{V}_+$ (understood either as active transformations or passive relabellings) are symmetry transformations for Bosons, but not for Fermions (and vice versa for $\mathcal{V}_-$ and Fermions). Similarly, for other types of systems without fundamental permutation invariance, e.g.\ anyons, none of the two groups would be symmetry transformations, and our results do not apply.

\section*{Discussion}
\refstepcounter{section}\label{main:sec:ConsequencesAndOutlook}

One of the goals of the QRF research program is to obtain novel physical predictions by postulating some version of covariance of the physical laws under QRF transformations~\cite{GiacominiBrukner2020,GiacominiBrukner2022,QRFIndefiniteMetric,Kabel2024}, such as the one formulated in~\cite{QRFIndefiniteMetric}: \textit{``Physical laws retain their form under quantum coordinate transformations.''} (for details see Subsection \hyperref[supp:sec:ImplicationsQuantCovPrinciples]{J} in Supplementary Note 5). This extended covariance principle in the context of quantum gravity is but one example of a larger class of quantum generalizations of symmetry principles considered in the literature, including also, for example, proposals for a quantum version of the equivalence principle~\cite{HardyQEP,GiacominiBrukner2020,GiacominiBrukner2022}, or a notion of quantum conformal symmetries~\cite{Kabel2024}. Our result is an instance of this: postulating invariance under quantum permutations for fundamental particles predicts the empirically correct absence of parastatistics. While this supports the idea that such principles can be predictively powerful, it also motivates some caution: we show that whether or not a given quantum permutation is a QRF transformation depends on the type of system (here Bosonic or Fermionic) it acts on. Therefore, only \textit{some} transformations will in general preserve the physical predictions for a given system, and we suggest the following modification of the postulate:

Physical laws retain their form under a suitable representation of the quantum coordinate transformation group.

Our results rule out fundamental parastatistics defined by multidimensional representations of the symmetric group. They do not exclude the possibility of emergent parastatistics, and hence do not stand in opposition to recent results showing the detectability of parastatistical correlations in many-body and quasiparticle systems \cite{Toppan2021Parabosons,Wang2025, WangHazzard}. In Example \ref{ExEmergent} in Supplementary Note 2 and Subsection \hyperref[SubsecFailure]{D} in Supplementary Note 3, we show that the intuitive reason for this is that permutations $U(\pi)$ on the emergent system $S$ are typically accompanied by non-trivial permutation actions on the complement of $S$, violating both our stronger notions of invariance in a natural way.

In principle, our mathematical results apply to all unitary representations of the symmetric group, enforcing that $U(\pi)$ acts either trivially or with the sign representation on the physical states. If $U(\pi)$ simply permutes the tensor factors or modes (as assumed above), then the corresponding symmetric and antisymmetric subspaces describe Bosons and Fermions, such that our results rule out parastatistics. However, if $U(\pi)$ is a more exotic representation that e.g.\ couples the permutations to internal degrees of freedom, then these allowed states may still have an interpretation as particles with exotic exchange statistics. For example, this happens with $\mathbb{Z}_2\times\mathbb{Z}_2$-graded parastatistics~\cite{ToppanFirst}, where creation operators are exchanged with a grading-dependent phase $\hat a^\dagger \hat b^\dagger\mapsto (-1)^{\langle \alpha,\beta\rangle} \hat b^\dagger \hat a^\dagger$. Such systems transform trivially under exotic permutation actions and are therefore consistent with our principles of quantum permutation invariance and complete invariance as applied to those exotic permutation actions.

Our work raises a number of interesting follow-up questions. May complete invariance point towards a deeper principle related to composing subsystems consistently, in the presence of more abstract, structural properties? Can further transformation properties of fundamental particles be understood as consequences of quantum covariance principles? These future questions notwithstanding, we believe that our results provide a novel perspective on the nature of Bosons and Fermions and thus of the very building blocks of our universe.

\section*{Methods}
\label{methods:sec:methods}

\subsection*{Representation theory and permutation-invariance}
\refstepcounter{section}\label{methods:sec:RepTheoryPermInv}

Suppose that we have $N\geq 2$ and a complex, separable Hilbert space $\mathcal{H}_1$ describing particles or modes, depending on which quantization formalism we use. We assume that the dimension of $\mathcal{H}_1$ is at least two, but at most countably-infinite, thus the total Hilbert space is $\mathcal{H}=\mathcal{H}_1^{\otimes N}$. We then have some unitary representation $\pi\mapsto U(\pi)$ of the permutation group $S_N$ whose irreps are labelled by Young diagrams (also called Young frames~\cite{Simon}) $\lambda$. The two special cases are
\begin{align}
\label{methods:eq:YoungDiag}
\ytableausetup
{mathmode, boxframe=normal, boxsize=1em}
\lambda_{\rm Bos}=\underbrace{\ydiagram{2}\ldots\ydiagram{1}}_{\textstyle N\mbox{ boxes}}
\quad \mbox{ and }\quad \lambda_{\rm Ferm}=\left.\begin{array}{c}\ydiagram{1,1}\\ \vdots \\ \ydiagram{1}\end{array}\right\} N \mbox{ boxes}
\end{align}
denoting one-dimensional representations of $S_N$: the trivial representation and the sign representation, respectively. These correspond to Bosons and Fermions. All other  Young frames $\lambda$ correspond to irreps that are at least two-dimensional, and these describe irreps of $S_N$ associated with paraparticles~\cite{Green,Peres}. Since $S_N$ is a finite group, its representation decomposes as
\begin{equation}
U(\pi)=\bigoplus_\lambda U_\lambda(\pi)\otimes \I_{n_\lambda}, 
\label{methods:eq:UDecomp}
\end{equation}
with the Hilbert space decomposing as
\begin{equation}
\mathcal{H}=\bigoplus_\lambda \mathcal{M}_\lambda\otimes\mathcal{N}_\lambda.
\label{methods:eq:DecompHS}
\end{equation}
Here, $n_\lambda=\dim(\mathcal{N}_\lambda)$ (which may be infinite) denotes the number of copies of the irrep $\lambda$ and $U_\lambda(\pi)$ the representation matrix on the irrep subspace $\mathcal{M}_\lambda$. Since these abstract decompositions apply to all unitary representations of $S_N$ on separable Hilbert spaces, our results will apply broadly to all of them.

Let us now postulate permutation-invariance: only quantum states $\rho$ with $U(\pi)\rho U(\pi)^\dagger=\rho$ are allowed to describe the results of physical preparation procedures. Hence, quantum states must be described by density operators that commute with all $U(\pi)$ and are of the form
\begin{equation}
\rho=\bigoplus_{\lambda}p_\lambda\frac{\I_{\lambda}}{d_{\lambda}}\otimes \rho_{\lambda},
\label{methods:eq:Superselection}
\end{equation}
where $d_\lambda$ is the dimension of the irrep $\lambda$, $\{p_\lambda\}_{\lambda}$ is a probability distribution, and the $\rho_\lambda$ are density matrices. All decompositions above are the direct result of Lemma \ref{lem:direct_sum_tensor_isomorphism} in Supplementary Note 1.

\subsection*{Formalism of complete invariance}
\refstepcounter{section}\label{methods:sec:FormalismCI}

Complete invariance dictates that ``physics is invariant'' under local permutations of $S$, even in the presence of an arbitrary ancillary system. Given a permutation-invariant state $\rho_S$ on a system $S$ of indistinguishable particles, it is expressed by the condition that all possible extensions $\rho_{SA}$ with $\rho_S=\Tr(\rho_{SA})$ are invariant under local permutations, i.e. 
\begin{equation}
\big (U(\pi)_S\otimes\I_A \big ) \rho_{SA} \big (U(\pi)_S^\dagger\otimes\I_A\big ) = \rho_{SA},
\label{methods:eq:CompleteInvariance}
\end{equation}
for all permutations $\pi$ any ancillary system $A$.

The following lemma shows that this notion can be mathematically defined in different ways.
\begin{lemma}
\label{methods:lem:LemCompleteInvariance}
The following conditions are all equivalent, and can be used to define what it means that a quantum state $\rho_S$ is \textbf{completely invariant} under permutations:
\begin{itemize}
    \item[(i)] Eq.(\ref{methods:eq:CompleteInvariance}) holds for some purification $\rho_{SA}$ of $\rho_S$;
    \item[(ii)] Eq.(\ref{methods:eq:CompleteInvariance}) holds for all purifications $\rho_{SA}$ of $\rho_S$;
    \item[(iii)] Eq.(\ref{methods:eq:CompleteInvariance}) holds for all extensions $\rho_{SA}$ of $\rho_S$.
\end{itemize}
Moreover, all three are equivalent to
\begin{itemize}
    \item[(iv)] $\rho_S$ has full support on the $\lambda_{\rm Bos}$- or $\lambda_{\rm Ferm}$-subspace.
\end{itemize}
\end{lemma}
\begin{proof}
Clearly, $(iii)\Rightarrow(ii)\Rightarrow(i)$, and we will now show that $(i)\Rightarrow(iv)$. We have the decomposition
\[
   U(\pi)_S\otimes\I_A = \bigoplus_\lambda U_\lambda(\pi)\otimes(\I_{n_\lambda}\otimes\I_A),
\]
and so Lemma \ref{lem:direct_sum_tensor_isomorphism} in Supplementary Note 1 implies that $\rho_{SA}=\bigoplus_\lambda p_\lambda \frac{\I_\lambda}{d_\lambda} \otimes\rho_{\lambda,A}$, where every $\rho_{\lambda,A}$ is a quantum state on $\mathcal{N}_\lambda\otimes A$. But if $\rho_{SA}$ is a pure state, then $p_\lambda>0$ is only possible if $d_\lambda=1$, leaving the two possibilities $\rho_{SA}=|\psi\rangle\langle\psi|_{\lambda_{\rm Bos},A}$ or $\rho_{SA}=|\psi\rangle\langle\psi|_{\lambda_{\rm Ferm},A}$. Hence, $\rho_S$ has full support on the $\lambda$-subspace, where either $\lambda=\lambda_{\rm Bos}$ or $\lambda=\lambda_{\rm Ferm}$.

Let us finally show that (iv)$\Rightarrow$(iii). Let $P_\lambda$ be the orthogonal projector onto $\cm_\lambda\otimes\cn_\lambda$, and let $\rho_{SA}$ be any extension of $\rho_S$. We have $1=\Tr(P_\lambda\rho_S)=\Tr((P_\lambda\otimes \I_A)\rho_{SA})$, and thus $(P_\lambda\otimes\I_A)\rho_{SA}(P_\lambda\otimes\I_A)=\rho_{SA}$. Using that $U(\pi)_S P_\lambda={\rm sgn}_\lambda(\pi)P_\lambda$, where ${\rm sgn}_\lambda(\pi)={\rm sgn}(\pi)$ if $\lambda=\lambda_{\rm Ferm}$ and ${\rm sgn}_\lambda(\pi)=1$ if $\lambda=\lambda_{\rm Bos}$, we obtain
\begin{align*}
&U(\pi)_S\otimes\I_A\rho_{SA}U(\pi)_S^\dagger\otimes\I_A \\
&= (U(\pi)_S\otimes\I_A)(P_\lambda\otimes\I_A)\rho_{SA}(P_\lambda\otimes\I_A)(U(\pi)_S^\dagger\otimes\I_A)\\
&=\rho_{SA},
\end{align*}
and so $\rho_{SA}$ satisfies~(\ref{methods:eq:CompleteInvariance}).
\end{proof}
This implies Thm.~\ref{main:th:CompleteInvariance} that rules out parastatistics in first quantization.

\subsection*{Formalism of quantum permutations}
\refstepcounter{section}\label{methods:sec:FormalismQuantPerm}

Starting from the definition given in the main text, we can give a more explicit form of the quantum permutations. Recall the decomposition of the total Hilbert space $\mathcal{H}=\bigoplus_\lambda \mathcal{M}_\lambda\otimes\mathcal{N}_\lambda$, on which the permutations act as $U(\pi)=\bigoplus_\lambda U_\lambda(\pi)\otimes\I_\lambda$. By Schur's Lemma, projective measurements commuting with all those unitaries are of the form
\begin{equation}
\{\I_\lambda\otimes P_{\lambda,j}\}_{\lambda,j},
\label{methods:eq:ProjMeasurement}
\end{equation}
and this will define a maximal quantum permutation group if the $P_{\lambda,j}$ are all rank-one projectors. Such maximally  finely-controlled quantum permutations are hence of the form
\begin{align}
V(\vec{\pi})&=\sum_{\lambda,j} e^{i\theta(\pi_{\lambda,j})}U(\pi_{\lambda,j}) (\I_\lambda \otimes P_{\lambda,j}) \nonumber \\&= \sum_{\lambda,j} e^{i\theta(\pi_{\lambda,j})}U_\lambda(\pi_{\lambda,j}) \otimes  P_{\lambda,j}
\label{methods:eq:QuantumPermutation}
\end{align}
with $\Tr(P_{\lambda,j})=1$. Note that we always have at least two particles (or modes) with associated Hilbert spaces of dimensions at least two, and so the number $Q$ of possible $(\lambda,j)$ pairs is at least two. In what follows, we will hence always assume that $Q\geq 2$. In the case of infinite-dimensional Hilbert spaces, we will restrict our attention to finite-rank projections for simplicity.

In the following lemma, let us now prove the first part of our main theorem: the existence of exactly two quantum permutation groups.

\begin{lemma}\label{methods:lem:proj_SNQ_rep}
The map $\vec{\pi}\mapsto V(\vec{\pi})$ in Eq.~(\ref{methods:eq:QuantumPermutation}) defines a quantum permutation group if and only if either $e^{i\theta(\pi_{\lambda,j})}=e^{i\theta}$ or $e^{i\theta(\pi_{\lambda,j})}=e^{i\theta}{\rm sgn}(\pi_{\lambda,j})$ for some fixed $\theta$ (which we will in the following, without loss of generality, set to zero). In both cases, if we set $\theta$ to zero, it is not only a projective, but a linear representation of $S_N^Q$.
\end{lemma}
\begin{proof}
The map $\vec{\pi}\mapsto V(\vec{\pi})$ is a projective representation if and only if there are complex numbers $\omega(\vec{\sigma},\vec{\pi})$ such that $V(\vec{\sigma})V(\vec{\pi})=\omega(\vec{\sigma},\vec{\pi})V(\vec{\sigma}\vec{\pi})$. Direct calculation shows that this implies
\[
\omega(\vec{\sigma},\vec{\pi})=e^{i\theta(\sigma_{\lambda,j})}e^{i\theta(\pi_{\lambda,j})}e^{-i\theta(\sigma_{\lambda,j}\pi_{\lambda,j})}\mbox{ for all }\lambda,j.
\]
Suppose that $\vec{\sigma}'$ and $\vec{\pi}'$ are elements of $S_N^Q$ that agree with $\vec{\sigma}$ and $\vec{\pi}$ on at least one entry $(\lambda,j)$, i.e.\ there exists some $(\lambda,j)$ such that $\sigma_{\lambda,j}=\sigma'_{\lambda,j}$ and $\pi_{\lambda,j}=\pi'_{\lambda,j}$. Then it follows that $\omega(\vec{\sigma},\vec{\pi})=\omega(\vec{\sigma}',\vec{\pi}')$. Let $\vec{\sigma}'',\vec{\pi}''$ be arbitrary elements of $S_N^Q$. Then we can always find some pair $\vec{\sigma}',\vec{\pi}'$ that agree with the pair $\vec{\sigma},\vec{\pi}$ in at least one entry, and that also agrees with the pair $\vec{\sigma}'',\vec{\pi}''$ in at least one entry. Thus, $\omega(\vec{\sigma},\vec{\pi})=\omega(\vec{\sigma}'',\vec{\pi}'')=\omega$ is a constant that does not depend on $\vec{\sigma}$ or $\vec{\pi}$. The special case $\pi_{\lambda,j}=\I$ shows that $\omega=e^{i\theta}$ for $\theta:=\theta(\I)$, and hence $\vec{\pi}\mapsto e^{-i\theta}V(\vec{\pi})$ is a linear unitary representation. Furthermore, $\pi\mapsto e^{-i\theta}e^{i\theta(\pi)}$ is a linear one-dimensional representation of $S_N$, i.e.\ either the trivial or the sign representation. This shows that $V(\vec{\pi})$ is of the claimed form.
\end{proof}
Thus, for every choice of one-dimensional projectors $P_{\lambda,j}$ in the subspaces $\mathcal{M}_\lambda$, there are (up to a global phase) two maximal quantum permutation groups, as stated in Theorem~\ref{main:th:QuantPermInvRulesOutPara}: one where $e^{i\theta(\pi)}=1$ and thus
\begin{equation}
V_+(\vec{\pi})=\sum_{\lambda,j} U_\lambda(\pi_{\lambda,j})\otimes P_{\lambda,j},
\label{methods:eq:QPermutation1}
\end{equation}
and another one where $e^{i\theta(\pi)}={\rm sgn}(\pi)$ and thus
\begin{equation}
V_-(\vec{\pi})=\sum_{\lambda,j} {\rm sgn}(\pi_{\lambda,j})U_\lambda(\pi_{\lambda,j})\otimes P_{\lambda,j}.
\label{methods:eq:QPermutation2}
\end{equation}
Note that we would have obtained the exact same result if we had started with an even more general definition of quantum permutations. In contrast to Definition~\ref{main:def:QuantumPerm} where the representation $V$ of the symmetric group is fixed, we could even let the representation (and thus the associated complex phases) depend on the permutation-invariant variable, i.e.\ on the branch $j$,
\begin{align}
    V(\vec{\pi}):= \sum_j e^{i \theta_j(\pi_j)} U(\pi_j) P_j.
    \label{methods:eq:MoreGeneral}
\end{align}
In Subsection \hyperref[supp:sec:moreGenQP]{G} in Supplementary Note 4, we show that this more general definition leads to the same conclusions as Definition~\ref{main:def:QuantumPerm}.

Now, in order to derive the form of states invariant under quantum permutations and thus to show the rest of our main result Theorem \ref{main:th:QuantPermInvRulesOutPara}, we will first need to establish the in-equivalence of representations on invariant subspaces for quantum permutations in Lemma (\ref{methods:lem:inequivReps}) below. 

Recall the decomposition of the Hilbert space~(\ref{methods:eq:DecompHS}), and rewrite it slightly as
\begin{equation}
\mathcal{H}=\bigoplus_\lambda\bigoplus_j \mathcal{H}_{\lambda,j},
\label{methods:eq:DecompDifferently}
\end{equation}
where every $\mathcal{H}_{\lambda,j}$ is an irreducible subspace for the representation $U(\pi)$ of~(\ref{methods:eq:UDecomp}); concretely, $\mathcal{H}_{\lambda,j}={\rm im}(\I_\lambda\otimes P_{\lambda,j})$. It is immediate from~(\ref{methods:eq:QPermutation1}) and~(\ref{methods:eq:QPermutation2}) that the $\mathcal{H}_{\lambda,j}$ are invariant subspaces for the quantum permutation groups $V(\vec{\pi})$. Since subspaces invariant under all $V(\vec{\pi})$ must also be invariant under all $U(\pi)=e^{-i\theta(\pi)}V(\pi,\ldots,\pi)$, the spaces $\mathcal{H}_{\lambda,j}$ cannot be decomposed any further into smaller-dimensional irreducible subspaces.

Thus, for both representations $\pi\mapsto U(\pi)$ and $\vec{\pi}\mapsto V(\vec{\pi})$, we have an identical decomposition~(\ref{methods:eq:DecompDifferently}) into irreducible subspaces. However, for the former, every pair of subspaces $\mathcal{H}_{\lambda,j}$ and $\mathcal{H}_{\lambda,k}$ carries equivalent representations of the permutation group, which is not true for the group of quantum permutations:
\begin{lemma}
\label{methods:lem:inequivReps}
If $\lambda\neq \mu$, then $\mathcal{H}_{\lambda,j}$ and $\mathcal{H}_{\mu,k}$ carry inequivalent representations of the quantum permutation group $V(\vec{\pi})\simeq S_N^Q$. Moreover, if $j\neq k$, then $\mathcal{H}_{\lambda,j}$ and $\mathcal{H}_{\lambda,k}$ carry inequivalent representations of $S_N^Q$ in all cases except for the following two:
\begin{itemize}
\item $\lambda=\lambda_{\rm Bos}$ and $V(\vec{\pi})=V_+(\vec{\pi})$,
\item $\lambda=\lambda_{\rm Ferm}$ and $V(\vec{\pi})=V_-(\vec{\pi})$.
\end{itemize}
\end{lemma}
\begin{proof}
Denote the action of $V(\vec{\pi})$ on the invariant subspace $\mathcal{H}_{\lambda,j}$ by $V_{\lambda,j}(\vec{\pi})$. Suppose that $\mu\neq\lambda$, but that the representations $\vec{\pi}\mapsto V_{\lambda,j}(\vec{\pi})$ and $\vec{\pi}\mapsto V_{\mu,k}(\vec{\pi})$ are equivalent. Then there is a unitary $W$ such that $V_{\lambda,j}(\vec{\pi})=W V_{\mu,k}(\vec{\pi}) W^\dagger$. The special case $\vec{\pi}=(\pi,\ldots,\pi)$ implies $U_{\lambda}(\pi)=W U_\mu(\pi)W^\dagger$, and this can only hold for all $\pi\in S_N$ if $\lambda=\mu$, which is a contradiction.

Now let $j\neq k$, and suppose that $\mathcal{H}_{\lambda,j}$ and $\mathcal{H}_{\lambda,k}$ carry equivalent representations of $S_N^Q$. Then there is some unitary $W$ such that $V_{\lambda,j}(\vec{\pi})=W V_{\lambda,k}(\vec{\pi})W^\dagger$ for all $\vec{\pi}$,  or equivalently, $\omega(\pi_j)U_\lambda(\pi_j)=\omega(\pi_k)W U_\lambda(\pi_k)W^\dagger$ for all $\pi_j,\pi_k\in S_N$, where by~\Cref{methods:lem:proj_SNQ_rep} $\omega(\pi)=1$ for all $\pi$ if $V(\vec{\pi})=V_+(\vec{\pi})$, or $\omega(\pi)={\rm sgn}(\pi)$ for all $\pi$ if $V(\vec{\pi})=V_-(\vec{\pi})$. For the special case $\pi_k=\I$, this implies $U_\lambda(\pi)=\omega(\pi)\I$ for all $\pi\in S_N$, and so $U_\lambda$ corresponds to the trivial representation acting on the Bosonic subspace in the case $V(\vec{\pi})=V_+(\vec{\pi})$, or the sign representation acting on the Fermionic subspace in the case $V(\vec{\pi})=V_-(\vec{\pi})$. In both cases, $U_\lambda$ is a one-dimensional representation acting on the Bosonic/Fermionic subspace respectively.
\end{proof}
Let us now analyze the consequences of quantum permutation invariance. We will denote the $N$ indistinguishable particles by $S$ (the system), and a possible additional quantum system (the ancilla) by $A$.  To consider $S$ in isolation, we can simply treat $A$ as a trivial quantum system of Hilbert space dimension one. Let us note that, unlike complete invariance, we are not imposing that the global invariant state is pure (or equivalently that every extension must be invariant). Rather, our goal is to determine the set of all invariant extensions.

First, let us consider the case that the quantum permutation group is the one of Eq.~(\ref{methods:eq:QPermutation1}), i.e.\  $V(\vec{\pi})=V_+(\vec{\pi})$. If we apply some quantum permutation on $S$ and the identity map on $A$, the resulting transformation can be decomposed as
\begin{equation}
V_{+}(\vec{\pi})_S\otimes\I_A=(\I_{\lambda_{\rm Bos}}\otimes \I_A)\oplus\bigoplus_{j,\lambda\neq\lambda_{\rm Bos}} V_{\lambda,j}(\vec{\pi})\otimes\I_A,
\label{methods:eq:QuantumPermutationFirstCase}
\end{equation}
and all the representations $\vec{\pi}\mapsto V_{\lambda,j}(\vec{\pi})$ in this equation are irreducible and pairwise inequivalent. Let us now implement the invariance of all physical predictions under $V_+(\vec{\pi})$ on the level of states, and postulate that all allowed states are invariant under all these quantum coordinate transformations. Using Schur's Lemma again (see Lemma \ref{lem:direct_sum_tensor_isomorphism} in Supplementary Note 1), the requirement $[V_+(\vec{\pi})_S\otimes\I_A,\rho_{SA}]=0$ for all $\vec{\pi}\in S_N^Q$ implies
\begin{align}
\label{methods:eq:QPermInvStatesFirstCase}
\rho_{SA}=p_{\rm Bos}\, \rho_{{\rm Bos},A}\oplus\bigoplus_{j,\lambda\neq\lambda_{\rm Bos}} p_{\lambda,j} \frac{\I_{\lambda,j}}{d_\lambda} \otimes \rho_A^{(\lambda,j)},
\end{align}
where $d_{\lambda}$ is the dimension of $\mathcal{H}_{\lambda,j}$ (which does not depend on $j$).
Furthermore, $p_{\rm Bos}$ and the $p_{\lambda,j}$ are probabilities that sum to one, $\rho_{\rm Bos,A}$ is a quantum state on $\mathcal{H}_{\rm Bos}\otimes A$, with $\mathcal{H}_{\rm Bos}$ the symmetric subspace of $S$, and the $\rho_A^{(\lambda,j)}$ are quantum states on $A$. This state describes a classical mixture of different types of statistics (for more details on how the $\lambda$ are related to types of paraparticles such as parabosons and parafermions, see Supplementary Note \hyperref[supp:sec:DescPara]{2}). If the quantum number is measured and found to be $\lambda=\lambda_{\rm Bos}$, then the post-measurement state $\rho_{\rm Bos,A}$ can describe any quantum state whatsoever, and it may in particular be pure and entangled. However, if another type of statistics $\lambda\neq \lambda_{\rm Bos}$ is found, then the post-measurement state will be
\begin{equation}
\rho_{SA}(\lambda)=\bigoplus_j p_j \frac{\I_{\lambda,j}}{d_\lambda}\otimes \rho_A^{(\lambda,j)},
\label{methods:eq:RhoSALambda}
\end{equation}
where $p_j=p_{\lambda,j}/\sum_j p_{\lambda,j}$. This is a separable state.

\begin{table*}[t]
    \centering
    \begin{tabular}{!{\vrule width 2pt}c!{\vrule width 2pt}c|c!{\vrule width 2pt}}
    \noalign{\hrule height 2pt}
    \hline
        Symmetry $V(g)$ & $\rho_S = V(g) \rho_S V(g)^\dagger$&   $\rho_{SA}$ pure, $\rho_{SA} = (V(g) \otimes \I_A) \rho_{SA} (V(g) \otimes \I_A)^\dagger$  \\
        \noalign{\hrule height 2pt}
        \hline
        $U(\pi), \pi \in S_N$ & $ \rho_S = \bigoplus_\lambda p_\lambda \frac{\I}{d_\lambda} \otimes \rho_\lambda$ &  $ \rho_S = \rho_\lambda$, $\lambda \in \{\lambda_{\rm Bos}, \lambda_{\rm Ferm}\}$ \\\hline
        $V(\vec{\pi}), \vec{\pi} \in S_N^{ Q}$ & $\rho_S = p_\gamma \rho_\gamma + \bigoplus_{\lambda \neq \gamma} p_\lambda \frac{\I}{d_\lambda} \otimes \rho_\lambda^{\rm diag}$, $\gamma \in \{\lambda_{\rm Bos}, \lambda_{\rm Ferm}\}$ &  $ \rho_S = \rho_\lambda$, $\lambda \in \{\lambda_{\rm Bos}, \lambda_{\rm Ferm}\}$ \\\hline
         $V(\vec{\pi};P), \vec{\pi} \in S_N^Q, P \in {\rm Proj}_{S_N}$ & $ \rho_S = p_\gamma \rho_\gamma + \bigoplus_{\lambda \neq \gamma} p_\lambda \frac{\I}{D_\lambda}$, $\gamma \in \{\lambda_{\rm Bos}, \lambda_{\rm Ferm}\}$  &  $ \rho_S = \rho_\lambda$, $\lambda \in \{\lambda_{\rm Bos}, \lambda_{\rm Ferm}\}$ \\\noalign{\hrule height 2pt}\hline
    \end{tabular}
    \caption{General form of invariant and completely invariant states for the different symmetries $\{U(\pi)\} \subset \{V(\vec{\pi})\} \subset \bigcup_P \{V(\vec{\pi}; P)\}$ considered in this paper, where $V(\vec{\pi};P)$ denotes either $V_+(\vec{\pi})$ or $V_-(\vec{\pi})$ defined relative to a projective measurement $P=\{P_j\}$, and $V(\vec{\pi})=V(\vec{\pi},P)$ for some fixed maximal $P$. We use the notation $\rho_\lambda^{\rm diag}=\sum_i p_{i,\lambda}P_{i,\lambda}$ for a quantum state diagonal in the basis of the measurement $\{P_{i,\lambda}\}$ controlling the quantum permutation, and ${\rm Proj}_{S_N}$ for the set of projective measurements where all projections commute with the representation of $S_N$. This table shows that without appeal to an ancillary system, one can effectively rule out parastatistics by appealing to invariance under quantum permutations $V(\vec{\pi})$ or $V(\vec{\pi};P)$. By appealing to an ancilla and the existence of a pure global state $\rho_{SA}$, invariance under the standard permutation group is enough to rule out parastatistics.}
    \label{methods:tab:summarytab1}
  \centering\bigskip
  \renewcommand{\arraystretch}{1.2}
  \begin{tabular}{!{\vrule width 2pt}c!{\vrule width 2pt}c|c|c!{\vrule width 2pt}}
  \noalign{\hrule height 2pt}
    \hline
    \multirow{2}{4cm}{ \centering Symmetry $V_S(g) \otimes \I_A$} &   Correlations with $\ch_A$ for $\lambda \in \{\lambda_{\rm Bos}, \lambda_{\rm Ferm}\}$ & \multicolumn{2}{c|}{Correlations with $\ch_A$ for $\lambda \not \in \{\lambda_{\rm Bos}, \lambda_{\rm Ferm}\}$}\\
    \cline{2-4}
    &  $\rho_{SA}$ arbitrary & $\rho_{SA}$ pure &  $\rho_{SA}$ arbitrary \\ \noalign{\hrule height 2pt}\hline
    $U(\pi), \pi \in S_N$  & Entanglement possible & No global state  & Entanglement possible \\ \hline
     $V(\vec{\pi}), \vec{\pi} \in S_N^{ Q}$ & Entanglement possible & No global state  &  Classical correlations only \\ \hline
     $V(\vec{\pi};P), \vec{\pi} \in S_N^Q, P \in {\rm Proj}_{S_N}$ & Entanglement possible & No global state  &  No correlations  \\\noalign{\hrule height 2pt} \hline
  \end{tabular}
   \caption{Allowed correlations with the environment when imposing invariance on system $S$ and ancilla $A$ for the different symmetries considered in this paper. Requiring that $\rho_{SA}$ is pure and invariant under the standard permutations $V(\pi)$ is enough to rule out parastatistics. If $\rho_{SA}$ is not required to be pure, then invariance under quantum permutations $\{V(\vec{\pi})\}$ or $\bigcup_P \{V(\vec{\pi},P)\}$ rules out parastatistics.}
   \label{methods:tab:summarytab2}
\end{table*}

However, we can say more. Due to~(\ref{methods:eq:RhoSALambda}), all possible $\rho_{SA}(\lambda)$ are diagonal in the same basis of $S$, which is determined by the projective measurement~(\ref{methods:eq:ProjMeasurement}). Thus, it is effectively a \textit{classical} system. It is well-known that classical systems cannot reversibly interact with quantum systems in a nontrivial way. While stochastic evolution of $SA$ under a Lindblad equation is possible~\cite{Oppenheim}, no unitary time evolution generated by any nontrivial Hamiltonian $H_{SA}\neq H_S+H_A$ is possible: any consistent coupling between a classical and a quantum system must be fundamentally irreversible~\cite{GGS}. Essentially, if the particle is not a Boson, it cannot interact reversibly with any ancillary system. (Up to the fineprint that paraboson Hilbert spaces carry a Bosonic sector which may be entangled with other systems, but which, as the wording indicates, behaves exactly like a Boson. For the corresponding parafermion case, see Example \ref{supp:ex:VolkovPara} in Supplementary Note 2).

The second case of the quantum permutation group (as in Eq.~\ref{methods:eq:QPermutation2})) can be treated analogously, starting from the decomposition
\begin{align}
\label{methods:eq:QuantumPermutationSecondCase}
V_-(\vec{\pi})_S\otimes\I_A=(\I_{\lambda_{\rm Ferm}}\otimes \I_A)\oplus\bigoplus_{j,\lambda\neq\lambda_{\rm Ferm}} V_{\lambda,j}(\vec{\pi})\otimes\I_A.
\end{align}
This establishes our main result.

Theorem~\ref{main:th:QuantPermInvRulesOutPara} also shows that classically, invariance of probability distributions under conditional permutations is equivalent to invariance under unconditional ones. This follows from the fact that Theorem~\ref{main:th:QuantPermInvRulesOutPara} does not yield any additional constraints in the case that all possible states and measurements are diagonal in the product basis.

If we assume the stronger form of quantum permutation invariance, we get an even stronger result:
\begin{theorem}
\label{methods:th:MainTheorem2}
Consider a system $S$ of indistinguishable particles, and assume strong quantum permutation invariance, i.e. invariance with respect to \textbf{all} maximal sets $\{P_{\lambda,j}\}$ of permutation-invariant projective measurements. Suppose that the particle type $\lambda$ and the quantum permutations $V(\vec{\pi})$ are among the following two cases:
\begin{itemize}
\item $\lambda=\lambda_{\rm Bos}$ and $V(\vec{\pi})=V_+(\vec{\pi})$ as in Eq.~(\ref{methods:eq:QPermutation1});
\item $\lambda=\lambda_{\rm Ferm}$ and $V(\vec{\pi})=V_-(\vec{\pi})$ as in Eq.~(\ref{methods:eq:QPermutation2}).
\end{itemize}
Then all quantum states of $S$ are allowed. Moreover, $S$ can be arbitrarily entangled with ancillary systems $A$.

In all other cases, however, and in particular if $\lambda\not\in\{\lambda_{\rm Bos},\lambda_{\rm Ferm}\}$, the state of $S$ must be maximally mixed:
\[
    \rho_S=\frac{\I_\lambda}{D_\lambda},
\]
where $D_\lambda=d_\lambda n_\lambda$ is the dimension of the paraparticle Hilbert subspace. Moreover, if we have an additional ancillary system $A$, then $S$ and $A$ must be uncorrelated:
\[
    \rho_{SA}=\frac{\I_\lambda}{D_\lambda}\otimes\rho_A,
\]
and no continuous interaction whatsoever between $S$ and $A$ is possible. If one of the multiplicity spaces is infinite-dimensional, i.e.\ if there is some $\lambda\not\in\{\lambda_{\rm Bos},\lambda_{\rm Ferm}\}$ with $n_\lambda=\infty$, then there does not exist any quantum state $\rho_S$ (or $\rho_{SA}$) with support on the corresponding subspace that satisfies strong quantum permutation invariance.
\end{theorem}
The proof of Theorem~\ref{methods:th:MainTheorem2} is given in Subsection \hyperref[supp:sec:proofTh3]{H} in Supplementary Note 4. Furthermore, Tables~\ref{methods:tab:summarytab1} and~\ref{methods:tab:summarytab2} below show the forms of invariant states for the different symmetries, i.e.\ under standard, quantum and strong quantum permutation-invariance, and the possible correlations with ancillary systems in each case.

\subsection*{Quantum reference frame descriptions}
\refstepcounter{section}\label{methods:sec:QRFstates}

The invariance of all physical predictions under permutations can be implemented in three different ways. First, we can have invariance on the level of states, by postulating that only quantum states $\rho$ with $U(\pi)\rho U(\pi)^\dagger=\rho$ are allowed to describe the results of preparation procedures, without any restriction on the observables. Second, we can implement invariance on the level of observables, requiring that only operators $A$ with $U^\dagger(\pi)A U(\pi)=A$ can be measured, without any restriction on the states. Third, we can implement invariance on the level of states and observables, requiring that both lie in the subalgebra of operators commuting with all $U(\pi)$. The three conventions are equivalent because
\begin{align}
\label{methods:eq:EquivInvConventions}
\Tr[\hat P(\rho)A]=\Tr[\rho \hat P(A)]=\Tr[\hat P(\rho)\hat P(A)]
\end{align}
holds, where $\hat P(X):=\frac 1 {|S_N|} \sum_{\pi\in S_N} U(\pi)X U(\pi)^\dagger$ projects every operator $X$ into the subalgebra of invariant operators. These three conventions are also available for other symmetries, including symmetry under quantum permutations.

Now, while we have been working with the first convention so far, the second one is sometimes chosen implicitly in the context of QRFs: only observables invariant under the symmetry are deemed measurable, but all quantum states are allowed as descriptions of preparation procedures. There is an algebra of observables that can be measured under a certain condition of invariance, for example the subalgebra of all Bosonic operators $\mathcal{A}_+$ together with invariance under the group of Bosonic quantum permutations $\mathcal{V}_+$. For any state $\rho \in \mathcal{S}(\mathcal{H})$, we then have $\Tr[\rho A_+]=\Tr[V_+(\vec{\pi})\rho V_+(\vec{\pi})^\dagger A_+]$ for all $V_+(\vec{\pi})\in \mathcal{V}_+$. This means that the set $\{V_+(\vec{\pi})\rho V_+(\vec{\pi})^\dagger|V_+(\vec{\pi})\in \mathcal{V}_+\}$ contains alternative, equally valid descriptions of one and the same quantum state. Choosing one description over another can be interpreted as choosing a quantum coordinate system. Since quantum permutation in $\mathcal{V}_+$ map from one quantum coordinate system to another, they are viewed as QRF transformations for Bosons. In Subsection \hyperref[supp:sec:QRFsApplBosFerm]{I} in Supplementary Note 5 we provide a short introduction and more details on QRFs. In~\cite{KrummHoehnMueller,HoehnKrummMueller}, it is demonstrated that the paradigmatic QRF transformations of ``jumping into the perspective of one of several particles'' under translation-invariance can be understood in exactly this way.

\section*{Data availability}
No datasets were generated or analyzed in this study.

\section*{Acknowledgments}

We are grateful to Andrea Di Biagio, Anne-Catherine de la Hamette, Borivoje Daki\'c, Bruna Sahdo, \v{C}aslav Brukner, David Gross, Esteban Castro-Ruiz, Hubert de Guise, Julian Maisriml, Lionel J.\ Dmello, Marina Maciel Ansanelli, Nicol\'{a}s Medina S\'{a}nchez, Robert W.\ Spekkens, Tanmay Singal, Tristan Malleville, Urs Schreiber, Viktoria Kabel, Y\`il\`e Y\={i}ng, Zhanna Kuznetsova and Zhiyuan Wang for helpful and stimulating discussions. We would like to express our particular gratitude to Francesco Toppan for insightful exchanges that helped us understand the relation of our results to $\Zl_2 \times \Zl_2$ graded parastatistics.

\section*{Author contributions}
M.P.M.\ designed the research project, based on calculations and ideas that have developed in joint discussions with T.G.\ and M.M. All three authors contributed to the proofs and calculations and to the preparation of the manuscript.


\section*{Funding}
Funded in whole or in part by the Austrian Science Fund (FWF) 10.55776/COE1 (Quantum Science Austria) and the European Union - NextGenerationEU.

\onecolumngrid

\newpage

\section*{Supplementary Note 1 -- Consequences of Schur's Lemma}

In the following we consider unitary representations of compact (but not necessarily Lie) groups $G$ on separable Hilbert spaces. We will always assume that $G$ is, as a topological space, a Hausdorff space.
We describe a consequence of Schur's Lemma~\cite[Proposition 5.8]{murnaghan} that is used several times in the main text. We make use of the fact that unitary representations of compact groups are isomorphic to the Hilbert space direct sum of irreducible representations~\cite[Theorem 7.8]{robert}, and irreducible representations of compact groups are finite-dimensional~\cite[Corollary 5.8]{robert}.

The following is well-known, but we include the proof in our notation for completeness. For the basic representation-theoretic facts that we are using, we refer the reader to the books by Barry Simon~\cite{Simon} and Alain Robert~\cite{robert}. Our two main cases of interest are the permutations and the quantum permutations, i.e.\ $\mathcal{G}=S_N$ and $\mathcal{G}=S_N^Q$, where $Q$ denotes the number of branches, i.e.\ independent copies of $S_N$. If $Q$ is not finite, but countably-infinite, then $S_N^Q$ is still a compact group in the product topology due to Tychonoff's Theorem. Moreover, it is a totally disconnected group and hence topologically a Hausdorff space~\cite{murnaghan}, i.e.\ a profinite group. The representation-theoretic results below hence apply to the case of $\mathcal{G}=S_N^\infty$, too.

\begin{lemma}\label{lem:direct_sum_tensor_isomorphism}
    Consider the unitary representation $U(g) = \bigoplus_{i\in N} U^i_\lambda(g)$, with $N$ a discrete set, of a compact group $G$ acting on $\ch \simeq \bigoplus_{i \in N} \ch_\lambda^i$, a direct sum of finite-dimensional Hilbert spaces $\ch_\lambda^i$, where each irreducible component is isomorphic, then we have the following (group representation) isomorphism:
    \begin{align}
        \ch \simeq \bigoplus_{i\in N} \ch_\lambda^i \simeq \cm_\lambda \otimes \cn_\lambda , 
    \end{align}
    where $\cm_\lambda \simeq \ch_\lambda$ and $\cn_\lambda \simeq \ell^2(N)$ (which is isomorphic to $\Cl^n$ for $N$ finite). The representation $U(g)$ acts as:
\begin{align}
    U(g) = \bigoplus_{i \in N}  U^i_\lambda(g) \simeq U_\lambda(g) \otimes \I_{\cn_\lambda} .
\end{align}
\end{lemma}

\begin{proof}
    Denote $\phi_i$ the representation isomorphism $\ch_\lambda^i \to \cm_\lambda$, where $\cm_\lambda$ carries the representation $U_\lambda(g)$, which exists since all $\ch_\lambda^i$ are isomorphic. Define $\Phi_i: \ch_\lambda^i \to \cm_\lambda \otimes \ch_i$ for $i \in N$ where $\ch_i \simeq \Cl$:
    \begin{align}
        \Phi_i: v^i \mapsto \phi_i(v^i) \otimes e_i,
    \end{align}
    where $v^i \in \ch_\lambda^i$ and $e_i$ a fixed unit vector in $\Cl$. The map $\Phi_i$ is a group representation isomorphism: $\Phi_i(U_\lambda^i(g) \bullet) = (U_\lambda(g) \otimes 1_i) \Phi_i(\bullet)$. Let $\Phi = \bigoplus_{i\in N}  \Phi_i: \ch \to \bigoplus_{i\in N}  (\cm_\lambda \otimes \ch_i)$. This map is also a group representation isomorphism. Moreover,
    \begin{align}
        \bigoplus_{i\in N}  (\cm_\lambda \otimes \ch_i) \simeq \cm_\lambda \otimes \bigoplus_{i\in N}  \ch_i = \cm_\lambda \otimes \cn_\lambda
    \end{align}
    where $\cn_\lambda =\ell^2(N)$. The action of $U(g)$ under this isomorphism is $U_\lambda(g) \otimes \I_{\cn_\lambda}$.
\end{proof}

\begin{lemma}
\label{LemInvarOperators}
    Given a unitary reducible representation $U(g) = \bigoplus_{i = 1}^n U^i_\lambda(g)$ acting on $\ch \simeq \bigoplus_{i = 1}^n \ch_\lambda^i$ a direct sum of separable Hilbert spaces $\ch_\lambda^i$, where each irreducible component is isomorphic, there exist unitary operators $W_{ij}:\ch_\lambda^i\to\ch_\lambda^j$ such that every operator $A$ with $[A,U(g)]=0$ can be written as
\begin{align}
    A = \sum_{ij} a_{ij} W_{ij}
\end{align}
for some $a_{ij}\in\mathbb{C}$. Moreover, the $W_{ij}$ are defined by the property $W_{ij}U_\lambda^j(g)=U_\lambda^i(g)W_{ij}$.
\end{lemma}

\begin{proof}
    Define the orthogonal projector $\Pi_i: \ch \to \ch_\lambda^i$. Then since $\sum_i \Pi_i = \I_{\ch}$ we can decompose $A$ as:
\begin{align}
    A = \sum_{i,j} \Pi_{i} A \Pi_j = \sum_{ij} A_{ij}.
\end{align}
This is a block decomposition of $A$, and $[A,U(g)]=0$ implies $A_{ij}U_\lambda^j(g)=U_\lambda^i(g)A_{ij}$. Hence, by Schur's Lemma, $A_{ij}$ is a multiple of a unitary operator $W_{ij}$ which is unique up to a constant which we can absorb into $a_{ij}$.
\end{proof}

\begin{lemma}\label{lem:inv_operator_tensor}
    Given a representation $U_\lambda(g) \otimes \I_{\cn_\lambda}$ acting on $\cm_\lambda \otimes \cn_\lambda$, any operator $A$ which commutes with $U_\lambda(g) \otimes \I_{\cn_\lambda}$ is of the form:
\begin{align}
    A = \I_{\cm_\lambda} \otimes A_{\cn_\lambda} . 
\end{align}
\end{lemma}
\begin{proof}
    We make use of the isomorphism $\phi: \bigoplus_i \ch_\lambda^i \to \cm_\lambda \otimes \cn_\lambda$. According to Lemma~\ref{LemInvarOperators}, we have
    \[
    W_{ki}W_{ij}U_\lambda^j(g)=W_{ki}U_\lambda^i(g)W_{ij}=U_\lambda^k(g)W_{ki}W_{ij}
    \]
    and thus $W_{ki}W_{ij}=W_{kj}$. Thus, we can find bases $\{|e_k,i\rangle\}_k$ of every $\ch_\lambda^i$ such that $W_{ij}|e_k,i\rangle=|e_k,j\rangle$. Hence,
    every invariant operator on $ \bigoplus_i \ch_\lambda^i$ is of the form:
    \begin{align}
        A = \sum_{ij} a_{ij} \sum_k \ketbra{e_k,i}{e_k,j}.
    \end{align}
    Under the isomorphism $\phi$ this is mapped to:
\begin{align}
    (\phi \otimes \phi^*)(A) &= \sum_{ij} a_{ij} \sum_k \ket{e_k} \ket i \bra{e_k}\bra{j} \\ &= \sum_k \ketbra{e_k}{e_k} \otimes \sum_{ij} a_{ij} \ketbra{i}{j} \\ &= \I_{\cm_\lambda} \otimes A_{\cn_\lambda} .
\end{align}
This shows that $A$ is of the form as claimed.
\end{proof}

\begin{lemma}
\label{supp:lem:LemSchurConsequence}
A unitary representation $U$ of a  compact group $G$ on a separable Hilbert vector space $\ch$ induces a decomposition:
\begin{align}
    \ch  \simeq \bigoplus_{\lambda} \cm_\lambda \otimes \cn_\lambda ,
\end{align}
where $\lambda$ labels the irreducible representations of $G$. The action of $U(g)$ is
\begin{align}
    U(g) \simeq  \bigoplus_{\lambda} U_{\cm_\lambda} \otimes \I_{\cn_\lambda}.
\end{align}
All invariant operators $A$ are of the form:
    \begin{align}
    A = \bigoplus_{\lambda}  \I_{\cm_\lambda} \otimes A_{\cn_\lambda} . 
\end{align}
\end{lemma}
\begin{proof}
Since $U(g)$ is a finite-dimensional representation of a compact group, it follows that it decomposes into a direct sum of irreducible representations:
        \begin{align}
        \ch \simeq \bigoplus_\lambda \bigoplus_{i_\lambda} \ch_{i_\lambda} , 
    \end{align}
Applying~\Cref{lem:direct_sum_tensor_isomorphism} to each subspace $\ch_\lambda:=\bigoplus_{i_\lambda}\ch_{i_\lambda}$ gives:
\begin{align}
    \ch  &\simeq \bigoplus_{\lambda} \cm_\lambda \otimes \cn_\lambda, \label{eqMultSpace16}\\
    U(g) &\simeq \bigoplus_{\lambda} U_\lambda(g) \otimes \I_{\cn_\lambda} . \label{eqMultSpace17}
\end{align}

An operator $A: \bigoplus \ch_\lambda \to \ch_\lambda$ can be decomposed as:
\begin{align}
    A = \sum_{\lambda, \lambda'} \sum_{i,j} \Pi_{\ch_{\lambda'}^j } A\Pi_{\ch_\lambda^i}  = \sum_{\lambda, \lambda' , i ,j} A_{\lambda, \lambda' , i ,j}, 
\end{align}
where each $A_{\lambda, \lambda' , i ,j} $ is an operator from $\ch_\lambda^i$ to $\ch_{\lambda'}^j$. 

Invariance $[A,U(g)]=0$ implies $A_{\lambda,\lambda',i,j} U_{\lambda'}^j(g)=U_\lambda^i(g)A_{\lambda,\lambda',i,j}$. For $\lambda \neq \lambda'$, Schur's Lemma entails that $A_{\lambda, \lambda' , i ,j} = 0$. Hence $A=\sum_{\lambda,i,j} A_{\lambda,\lambda,i,j}$ is block-diagonal in the $\ch_\lambda$, and so $A_\lambda:=\sum_{ij} A_{\lambda,\lambda,i,j}$ satisfies $[A_\lambda,U_\lambda(g)\otimes\I_{\mathcal{N}_\lambda}]=0$. Using~\Cref{lem:inv_operator_tensor} on each $\ch_\lambda \simeq \cm_\lambda \otimes \cn_\lambda$ gives:
\begin{align}
    A = \bigoplus_{\lambda}  \I_{\cm_\lambda} \otimes A_{\cn_\lambda} . 
\end{align}
This proves the main result of this section.
\end{proof}

\section*{Supplementary Note 2 -- Descriptions of parastatistics}
\label{supp:sec:DescPara}

\subsection{Examples and classification}

We first illustrate a Bosonic, Fermionic and paraparticle representation in first quantization for the case of $N=3$ particles, as has been discussed in~\cite{Peres}. The action of standard permutations in this context is given by 
\begin{equation}
U(\pi)|x_1,\ldots,x_N\rangle=|x_{\pi^{-1}(1)},\ldots,x_{\pi^{-1}(N)}\rangle,
\end{equation}
which permutes the tensor factors associated to Hilbert spaces describing single particles.

\begin{example}
\label{supp:ex:PeresExample}
    For simplicity, let us only consider states of the form $|x_1,x_2,x_3\rangle$  where $x_1,x_2$ and $x_3$ are all distinct and orthonormal. First, note that the states
\begin{align}
    |b\rangle &=\frac{1}{|S_3|}\sum_{\pi \in S_3}U(\pi)|x_1,x_2,x_3\rangle, \\
    |f\rangle &=\frac{1}{|S_3|}\sum_{\pi \in S_3}\textnormal{sgn}(\pi)U(\pi)|x_1,x_2,x_3\rangle
\end{align}
fulfill $U(\pi)|b\rangle=|b\rangle$ and $U(\pi)|f\rangle=\textnormal{sgn}(\pi)|f\rangle$ for all $\pi \in S_3$. Hence, they are contained in the symmetric (Bosonic) and antisymmetric (Fermionic) subspaces of $\mathcal{H}$, on which $U(\pi)$ acts as the trivial and the sign representation, respectively.

Now, for $\omega=e^{i\frac{2\pi}{3}}$, let us also consider the orthonormal states
\begin{align}
|\psi_1\rangle&=\frac{1}{\sqrt{3}}(|x_1,x_2,x_3\rangle +\omega|x_2,x_3,x_1\rangle+\bar{\omega}|x_3,x_1,x_2\rangle),\\
|\psi_2\rangle&=\frac{1}{\sqrt{3}}(|x_2,x_1,x_3\rangle +\omega|x_3,x_2,x_1\rangle+\bar{\omega}|x_1,x_3,x_2\rangle).
\end{align}
For the permutation $\tau=(12)$ exchanging positions 1 and 2, we have $U(\tau)|\psi_1\rangle=|\psi_2\rangle$, while the cyclic permutation $\sigma=(123)$ gives $U(\sigma)|\psi_1\rangle=\omega|\psi_1\rangle$ and $U(\sigma)|\psi_2\rangle=\bar\omega|\psi_2\rangle$. Therefore, since every $\pi\in S_3$ is generated by the elements $\tau$ and $\sigma$, we find that $\mathcal{H}_{\rm Para}:=\textnormal{span}\{|\psi_1\rangle,|\psi_2\rangle\}$ is a two-dimensional invariant subspace under $U(\pi)$, which acts as an irrep $U_{\rm Para}(\pi)$ on $\mathcal{H}_{\rm Para}$. This is the simplest example of a parastatistics sector. The corresponding Young diagram is given by
\begin{align}
\lambda_{\rm Para}=\ydiagram{2,1}\,\,.
\end{align}

Moreover, we observe that $U(\pi)$ acts as the regular representation $U_r(\pi)$ on the invariant subspace $\mathcal{H}_r:=\textnormal{span}\{U(\pi)|x_1,x_2,x_3\rangle|\pi \in S_3\}$, where  $\ket{x_1, x_2, x_3}$ is a fixed initial state,  which contains the states $|b\rangle,|f\rangle,|\psi_1\rangle$ and $|\psi_2\rangle$ described there. Therefore, we must have 
\begin{align}
    U_r(\pi)=1 \oplus \textnormal{sgn}(\pi)\oplus U_{\rm Para}(\pi)\otimes\I_2,
\end{align}
and thus the orthogonal decomposition
\begin{align}
    \mathcal{H}_r=\comp|b\rangle \oplus \comp|f\rangle\oplus2\mathcal{H}_{\rm Para},
\end{align}
as $1,\sgn(\pi)$ and $U_{\rm Para}(\pi)$ are the only irreps of $S_3$. Furthermore, this tells us that $U_{\rm Para}(\pi)$ must appear twice. Indeed, a second instance of $\mathcal{H}_{\rm Para}$ is spanned by the orthonormal pair
\begin{align}
|\phi_1\rangle&=\frac{1}{\sqrt{3}}(|x_2,x_1,x_3\rangle +\bar{\omega}|x_3,x_2,x_1\rangle+\omega|x_1,x_3,x_2\rangle),\\
|\phi_2\rangle&=\frac{1}{\sqrt{3}}(|x_1,x_2,x_3\rangle +\bar{\omega}|x_2,x_3,x_1\rangle+\omega|x_3,x_1,x_2\rangle),
\end{align}
for which we also find $\langle\psi_i|\phi_j\rangle=0$, due to $1+\omega+\bar{\omega}=0$.

Moreover, given $U(\tau)|\psi_1\rangle=|\psi_2\rangle$ and $U(\tau)|\psi_2\rangle=|\psi_1\rangle$ for $\tau=(12)$, the states are related by a permutation and cannot be distinguished due to our postulate of permutation-invariance. Hence, they lead to the same physical predictions as the mixed state $\rho_\psi=\frac{1}{2}(|\psi_1\rangle \langle\psi_1|+|\psi_2\rangle \langle\psi_2|)$. This state is invariant, and is hence a valid description of preparation procedures according to permutation-invariance on the level of states. This is true for the pair $|\phi_1\rangle,|\phi_2\rangle$ as well, leading to $\rho_\phi$ there. However, although each pair cannot be distinguished, it is possible to distinguish the two pairs from each other, since $\rho_\psi$ and $\rho_\phi$ are both invariant and orthogonal to one another.

\end{example}

Any collection of $N$ Bosons (Fermions) is described by states contained in irreducible subspaces that correspond to the Young diagram of a single row (column) of $N$ boxes. Moreover, let us consider a collection of $N+1$ Bosons (Fermions), $N$ of which are localized in a lab and one very far away, such that it cannot be accessed by an observer in the lab. It is well-known that the representation obtained by restricting a representation of the group $S_{N+1}$ with Young diagram $\lambda$ to the subgroup $S_N$ is the direct sum of all representations with Young diagram $\lambda'$  obtained by removing a square from $\lambda$~\cite{Simon}.  Hence, given the corresponding Young diagram of a single row (column) of $N+1$ boxes, the only valid Young diagram obtained by removing one box is the single row (column) of $N$ boxes. Therefore, it is clear that the $N$ particles in the lab must behave as Bosons (Fermions) as well. In the same way, by considering the para-sector of three particles as in Example \ref{supp:ex:PeresExample} above, we can imagine a scenario where we only have access to two of the particles, and no access whatsoever to the third particle (this includes not having access to the outcomes of potential measurements that some other observer might have performed on the third particle). Since the mere existence of the third particle far away should not influence  the measurement outcome statistics in the lab, the accessible pair should therefore either behave as a pair of Bosons or Fermions. 

This is indeed shown in~\cite{Taylor1} and generalized in their classification of paraparticles \cite{Taylor2}: parabosons (parafermions) of order $p\in \nats$ correspond to the family of Young diagrams with a maximum number of $p$ rows (columns). Restricting to a subset of the particles corresponds to removing boxes from the Young diagrams, and this preserves the particle type if it is defined in this way. Note that parafermions as used here should not be confused with the notion of parafermions in the context of non-Abelian anyons \cite{osti_5929972,Fendley_2014}.

\subsection{Parastatistics in second quantization}
\label{supp:sec:ParaStatSecQuant}

It is second quantization in which parastatistics was first conceived of by Green \cite{Green}. There, Bosons (Fermions) are described by a Fock space $\mathcal{F}$ generated by creation operators $a_{k}^\dagger$ that fulfill the usual (anti)commutation relations, where $k$ denote modes such as momentum or position. However, in order to preserve the equations of motion, one can have more general, trilinear (anti)commutation relations \cite{Green, Messiah_Greenberg}:
\begin{align}
\label{supp:eq:TrilRel}
    [a_{k_1},[a_{k_2}^\dagger,a_{k_3}]_{\pm}]_-&=2\delta_{k_1k_2}a_{k_3},\\ [a_{k_1},[a_{k_2}^\dagger,a^\dagger_{k_3}]_{\pm}]_-&=2\delta_{k_1k_2}a^\dagger_{k_3}\pm2\delta_{k_1k_3}a^\dagger_{k_2}, \\ [a_{k_1},[a_{k_2},a_{k_3}]_{\pm}]_-&=0,  
\end{align}

with $[X,Y]_\pm:=XY\pm YX$ denoting anticommutators and commutators, respectively. These relations led to a whole new class of fields that were eventually named parabosons (parafermions), corresponding to the choice of anticommutators (commutators) above. Green's ansatz of solutions for the trilinear equations involves writing the parabosonic (parafermionic) operators as sums of $p$ distinct Bosonic (Fermionic) creation and annihilation operators:
\begin{align}
    a^{(\dagger)}_k=\sum_{\alpha=1}^pc^{(\dagger)}_{\alpha,k}
\end{align}
with
\begin{align}
    [c^\dagger_{\alpha,k_1},c_{\alpha,k_2}]_\mp &= \delta_{k_1,k_2}, \hspace{0.5cm} [c_{\alpha,k_1},c_{\alpha,k_2}]_\mp=0 \nonumber \\ [c^\dagger_{\alpha,k_1},c_{\beta,k_2}]_\pm &= 0, \hspace{0.5cm} [c_{\alpha,k_1},c_{\beta,k_2}]_\pm=0, \hspace{0.5cm} \alpha \neq \beta
\end{align}
Parabosonic (parafermionic) operators $a^{(\dagger)}_k$ of order $p$ are then identified with sums of $p$ Bosonic (Fermionic) operators $c^{(\dagger)}_{\alpha,k}$ which commute (anticommute) for $\alpha=\beta$ and anticommute (commute) for $\alpha \neq \beta$. Moreover, it can be shown \cite{Messiah_Greenberg} that they exhaust all possible solutions compatible with the trilinear relations. In addition to the compatibility with the equations of motion, the trilinear relations can also be derived by requiring invariance of the number operator under unitary changes of basis~\cite{Bialynicki-Birula}.

The $N$-particle subspace of the Fock space is spanned by vectors
\begin{align}
    |k_1,...,k_N\rangle_\mathcal{F} := a_{k_1}^\dagger...a_{k_N}^\dagger|0\rangle_\mathcal{F},
\end{align}
where $|0\rangle_\mathcal{F}$ is the vacuum state. By fixing the entries and order of such a tuple $K:=(k_1,...,k_n)$, one can define a representation $\bar U_K(\pi)$ of $S_N$ on the subspace $\mathcal{H}_K:=\textnormal{span}\{|k_{\pi(1)},...,k_{\pi(N)}\rangle_\mathcal{F}|\pi \in S_N\}$, which transforms creation operators $a_{k_i}^\dagger$ into $\bar U_K(\pi) a_{k_i}^\dagger\bar U_K(\pi)^\dagger=a_{k_{\pi(i)}}^\dagger$. It can then be shown \cite{Ohnuki} that this representation decomposes in the following way: For
parabosons (parafermions) of order $p$, on each of these subspaces, the representation contains exactly once each irrep whose Young diagram has no more than $p$ rows (columns). Thus, a connection between the parafields of second quantization and the family of Young diagrams associated to each paraparticle type in first quantization is established. Defining $\bar U_K(\pi)$ for every ordered tuple $K$ leads to a consistent definition of standard permutations $\bar U(\pi)$ in second quantization, known as \textit{particle permutations} (for more details, see Subsection \ref{supp:sec:StandPermSecQuant}).

\begin{example}
\label{supp:ex:VolkovPara}
    Following the Volkov model from \cite{LandshoffStapp}, we illustrate the case for the simplest parafermion of order 2. There, a product of three creation operators fulfills the relation $a_i^\dagger a_j^\dagger a_k^\dagger=-a_k^\dagger a_j^\dagger a_i^\dagger$ for $i,j,k \in \{1,2,3\}$. This relation halves the dimension of the space of states created from the vacuum $|0\rangle_\mathcal{F}$ down to 3, and the correspondence described above tells us that it will consist of irrep spaces associated to Young diagrams with a maximum of 2 columns. Indeed, on the one hand we find the state
    \begin{align}
        |f\rangle_{\mathcal{F}}= \frac{1}{\sqrt{3}}\left(a^\dagger_{1}a^\dagger_{2}a^\dagger_{3}
    + a^\dagger_{2}a^\dagger_{3}a^\dagger_{1}+ a^\dagger_{3}a^\dagger_{1}a^\dagger_{2} \right)|0\rangle_\mathcal{F},
    \end{align}
    which exhibits usual Fermionic statistics, i.e. it carries the one-dimensional sign representation under particle permutations. On the other hand we find the two-dimensional irrep space spanned by the states
    \begin{align}
        |\psi_{1}\rangle_\mathcal{F} &= \frac{1}{\sqrt{3}}\left(a^\dagger_{1}a^\dagger_{2}a^\dagger_{3}
        + \bar\omega a^\dagger_{2}a^\dagger_{3}a^\dagger_{1}+  \omega a^\dagger_{3}a^\dagger_{1}a^\dagger_{2} \right)|0\rangle_\mathcal{F}, \\
        |\psi_{2}\rangle_\mathcal{F} &= \frac{1}{\sqrt{3}}\left(a^\dagger_{2}a^\dagger_{1}a^\dagger_{3}
        + \bar \omega a^\dagger_{1}a^\dagger_{3}a^\dagger_{2}+ \omega a^\dagger_{3}a^\dagger_{2}a^\dagger_{1} \right)|0\rangle_\mathcal{F}\nonumber \\=&-\frac{1}{\sqrt{3}}\left(a^\dagger_{3}a^\dagger_{1}a^\dagger_{2}
        + \bar \omega a^\dagger_{2}a^\dagger_{3}a^\dagger_{1}+ \omega a^\dagger_{1}a^\dagger_{2}a^\dagger_{3} \right)|0\rangle_\mathcal{F},
    \end{align}
    with $\omega=e^{i\frac{2\pi}{3}}$. As in first quantization, for the permutation $\tau=(12)$ we have $\bar U(\tau)|\psi_1\rangle_\mathcal{F}=|\psi_2\rangle_\mathcal{F}$, while the cyclic permutation $\sigma=(123)$ gives $\bar U(\sigma)|\psi_1\rangle_\mathcal{F}=\omega|\psi_1\rangle_\mathcal{F}$ and $\bar U(\sigma)|\psi_2\rangle_\mathcal{F}=\bar\omega|\psi_2\rangle_\mathcal{F}$. A permutation-invariant and therefore physically allowed state is given by
    \begin{align}
        \rho_S=\frac{1}{2}\left(|\psi_1\rangle\langle\psi_1|_\mathcal{F}+|\psi_2\rangle\langle\psi_2|_\mathcal{F} \right).
    \end{align}
\end{example}
It might seem at first sight that the postulate of permutation-invariance renders parastatistics trivial: in the previous example, it enforced a uniform mixture on the subspace spanned by $|\psi_1\rangle$ and $|\psi_2\rangle$, and one might conjecture that this ``averages out'' all non-trivial consequences of parastatistics. However, in Example~\ref{supp:ex:detectPara} below, we show that this is not the case, and that interesting permutation-invariant physical behavior can be obtained, even if we have a single particle in each mode. Allowing situations with more than one particle per mode shows this in an even simpler way. Consider three parafermions in two modes. There are two orthogonal particle-permutation-invariant states
\begin{align}
    \ket{\psi_\pm}:=( a_1^\dagger  a_1^\dagger  a_2^\dagger \pm   a_2^\dagger a_2^\dagger a_1^\dagger) \ket 0,
\end{align}
and the other possible permutation-invariant combinations are:
\begin{align}
    (a_1^\dagger  a_2^\dagger a_1^\dagger \pm  a_2^\dagger \hat a_1^\dagger  a_2^\dagger) \ket 0 = 0 \\
    ( a_1^\dagger  a_2^\dagger a_2^\dagger \pm  a_2^\dagger  a_1^\dagger  a_1^\dagger) \ket 0 = \mp \ket{\psi_\pm}
\end{align}

Since it is possible to have two modes with three excitations, these particles are not Fermions. Moreover, since it is not possible to have three excitations in one mode, they  also behave differently from Bosons.

We have just seen an example for parafermionic states in second quantization. For a given triple of distinct modes and under the postulate of (particle) permutation-invariance, there is one physically allowed state. However, one must note that the multiplicity of possible states arises from having other degrees of freedom available, i.e. more modes, on which we can condition. This is in no way different for a permutation-invariant state of a pair of Bosons and Fermions, when there are only two modes available. In the following example, we will introduce an additional spin degree of freedom and implement a conditional permutation to show that paraparticles are still detectable, as well as different from Bosons and Fermions despite imposing permutation-invariance. 

\begin{example}
\label{supp:ex:detectPara}
Recall Example~\ref{supp:ex:VolkovPara} involving three parafermions in three distinct modes, and consider an additional qubit with basis states $\{|0\rangle,|1\rangle\}$. This qubit may describe an arbitrary external degree of freedom, or it may (similarly as in the experimental proposal by Roos et al.~\cite{Roos}) describe internal degrees of freedom of the particles, for example $|0\rangle=|\uparrow\uparrow\uparrow\rangle$ if all three spins are up, or $|1\rangle=|\downarrow\downarrow\downarrow\rangle$ if all three spins are down. We start with the permutation-invariant initial product state
\begin{align}
    \rho&=\frac{1}{4}\left(|\psi_1 \rangle \langle\psi_1|_\mathcal{F}+|\psi_2 \rangle \langle\psi_2|_\mathcal{F} \right)\otimes \left(|0\rangle +|1\rangle \right) \left(\langle 0 | +\langle 1 | \right).
\end{align}
We will restrict our considerations to the subspace spanned by $|\psi_1\rangle$and $|\psi_2\rangle$, because $\rho$ is fully supported on it. Similarly as in the Roos et al.~\cite{Roos} experiment, we perform a permutation that is coherently controlled by the qubit: if the qubit is in state $0$, we do nothing; if it is in state $1$, we do a cyclic permutation. However, whether we apply the cyclic permutation $\sigma=(123)$ or $\sigma^{-1}$ shall also be controlled by whether the particles are in state $|\psi_1\rangle$ or $|\psi_2\rangle$. That is, we apply the unitary transformation
\begin{align}
    \bar V&=\mathbb{I}\left(|\psi_{1} \rangle \langle \psi_{1}|_\mathcal{F}\otimes|0\rangle \langle 0 |+|\psi_{2} \rangle \langle \psi_{2}|_\mathcal{F}\otimes|0\rangle \langle 0 | \right) \nonumber \\&+\bar U(\sigma)|\psi_{1} \rangle \langle \psi_{1}|_\mathcal{F}\otimes|1\rangle \langle 1 |+\bar U(\sigma^{-1})|\psi_{2} \rangle \langle \psi_{2}|_\mathcal{F}\otimes|1\rangle \langle 1|.
\end{align}
This is a conditional permutation, but it is not a quantum permutation in the sense of Definition~1 in the main text: the projectors $|\psi_i\rangle\langle\psi_i|_{\mathcal{F}}\otimes|1\rangle\langle 1|$ do not commute with all permutations. However, they \textit{do} commute with those permutations $U(\pi)$ where $\pi$ is in the cyclic group $\mathbb{Z}_3\simeq\{\mathbb{I},\sigma,\sigma^{-1}\}$. Hence, we can regard it as a valid $\mathbb{Z}_3$ QRF transformation. It maps the initial state to $\rho':=\bar V \rho \bar V^\dagger$, which is
\begin{align}
    \rho'&=\frac{1}{4}\left(|\psi_1 \rangle \langle\psi_1|_\mathcal{F}+|\psi_2 \rangle \langle\psi_2|_\mathcal{F} \right)\otimes \left(|0\rangle +\omega|1\rangle \right) \left(\langle 0 | +\bar\omega\langle 1 | \right),
\end{align}
i.e.\ it induces a relative phase of $\omega=\exp\left(\frac{2\pi}{3} i\right)$ on the control qubit. This is clearly impossible for Bosons and Fermions. To see that it has this effect, note that $\bar V$ can be simplified to
\begin{align}
   \bar V=\left(|\psi_1 \rangle \langle\psi_1|_\mathcal{F}+|\psi_2 \rangle \langle\psi_2|_\mathcal{F} \right)\otimes\left(\strut |0\rangle\langle 0|+\omega |1\rangle\langle 1|\right),
\end{align}
which also shows that it is indeed a permutation-invariant transformation, mapping every permutation-invariant state to a permutation-invariant state.

It may seem surprising that we should regard this as an allowed operation: after all, it performs an action depending on whether the internal degree of freedom is $|\psi_1\rangle$ or $|\psi_2\rangle$, and permutation-invariance forbids the measurement of the corresponding projectors. However, these projectors are not used for \textit{measurement}, but for \textit{coherent control} which can be thought of as fundamentally erasing the corresponding outcome in the process. Hence, particle-permutation-invariant states of three parafermions, occupying three distinct modes, can induce relative phases of $\omega=\exp\left(\frac{2\pi}{3} i\right)$ via physically allowed conditional permutations which Bosons and Fermions cannot.
\end{example}

Both Examples \ref{supp:ex:VolkovPara} and \ref{supp:ex:detectPara} can be understood in the same way for parabosons when the relation  $a_i^\dagger a_j^\dagger a_k^\dagger=a_k^\dagger a_j^\dagger a_i^\dagger$ is fulfilled by the creation operators. Both relations are instances of trilinear relations that are more general than ordinary (anti)commutation relations, corresponding to parabosons (parafermions) of the lowest non-trivial order 2. In the case of parafermions, the order coincides with the maximum occupation number per mode: in contrast to ordinary Fermions, $(a_i^\dagger)^2 \neq 0$, but the relations clearly imply that $(a_i^\dagger)^3 = 0$. However, in the special arrangement that the operators have in the Fermionic state $|f\rangle_\mathcal{F}$ in Example \ref{supp:ex:VolkovPara}, setting any two of them equal will lead to 0, as expected.     

\subsection{Example of emergent parastatistics}
In this subsection, we give an example (following~\cite[Section 9]{Landsman:2013ooo}) that shows how parafermions can emerge from Bosons with additional degrees of freedom. In Subsection \hyperref[SubsecFailure]{D} of Supplementary Note 3, we will see that it demonstrates an intuitive reason for why and how complete invariance fails for such emergent systems.
\begin{example}
\label{ExEmergent}
    We begin with three distinguishable isospin doublets $\ch^{(3)} = (L^2(\Rl^3) \otimes \Cl^2)^{\otimes 3}$, which carry the following representation of $S_3$:
\begin{align}\label{eq:S3_iso}
    U^{(3)}(\pi)\psi_{a_1, a_2,a_3}(q_1,q_2,q_3)=\hspace{0.1cm}\psi_{a_{\pi(1)}, a_{\pi(2)},a_{\pi(2)}}(q_{\pi(1)},q_{\pi(2)},q_{\pi(3)}).
\end{align}    
We consider the subspace $\ch^{(3)}_{\lambda_0}$ of Bosonic isospin doublets spanned by symmetrized wavefunctions:
\begin{align}
\label{eq:isospin_S3}
    \psi_{a_1, a_2,a_3}(q_1,q_2,q_3) = \psi_{a_{\pi(1)}, a_{\pi(2)},a_{\pi(2)}}(q_{\pi(1)},q_{\pi(2)},q_{\pi(3)}).
\end{align}

Using Schur-Weyl for the action of $S_3$ given in~\Cref{eq:S3_iso} (i.e. permuting both spatial and spin labels) we have that the space $\ch^{(3)}$ decomposes as:
\begin{align}
    \ch^{(3)} \simeq \bigoplus_{\lambda} V_\lambda \otimes W_\lambda \simeq \bigoplus_{\lambda} \ch^{(3)}_\lambda ,
\end{align}
where $V_\lambda$ carries the irreducible representation of $S_3$ labeled by $\lambda$ and $W_\lambda$ a representation of $U(W_\lambda)$. Since $W_\lambda$ is infinite-dimensional, one needs to choose the correct topology on $W_\lambda$ to define $U(W_\lambda)$, as discussed in \cite[footnote 7]{Landsman:2013ooo}.

However we can further decompose $ \ch^{(3)}_{\lambda_0}$ by considering the fact that the spatial and spin components of $\ch^{(3)}$ carry separate actions of $S_3$.
We write $\ch^{(3)}  \simeq \ch^{(3),\spatial} \otimes \ch^{(3),\spin} $ where  $\ch^{(3),\spatial} \simeq  L^2(\Rl^3)^{\otimes 3}$ and $\ch^{(3),\spin}  \simeq (\Cl^2)^{\otimes 3}$. 
$\ch^{(3),\spatial}$ carries a representation $U^{(3),\spatial}(\pi)$ of $S_3$ (corresponding to permuting the spatial indices) and similarly  $\ch^{(3),\spin}$ carries a representation $U^{(3),\spin}(\pi)$ of $S_3$ (corresponding to permuting the spin indices).

The representation $U^{(3)}$ of $S_3$  given in \Cref{eq:S3_iso} corresponds to performing the same permutation $\pi$ on both  $\ch^{(3),\spatial}$ and $\ch^{(3),\spin}$, i.e.   $U^{(3)}(\pi) = U^{(3),\spatial}(\pi) \otimes U^{(3),\spin}(\pi) $. 

Using Schur-Weyl we can first decompose the spaces  $\ch^{(3),\spatial}$ and $\ch^{(3),\spin}$  under their respective $S_3$ actions:
\begin{align}
    \ch^{(3),\spatial} = L^2(\Rl^3)^{\otimes 3} &\simeq \bigoplus_{\lambda} V_\lambda^\spatial \otimes W_\lambda^\spatial \simeq \bigoplus_{\lambda} \ch^{(3),\spatial}_\lambda , \\
    \ch^{(3),\spin} =  (\Cl^2)^{\otimes 3} &\simeq \bigoplus_{\lambda} V_\lambda^\spin \otimes W_\lambda^\spin \simeq \bigoplus_\lambda \ch^{(3),\spin}_\lambda,
\end{align}
where $V_\lambda^\spatial, V_\lambda^\spin$ carry the irreducible representation of $S_3$ labeled by $\lambda$, $W_\lambda^\spatial$ a representation of $U(W_\lambda^\spatial)$ and $W_\lambda^\spin$ an irreducible representation of ${\rm SU}(2)$. 

Clearly for $\lambda_0$ the Bosonic representation, any state in $\ch^{(3),\spatial}_{\lambda_0} \otimes \ch^{(3),\spin}_{\lambda_0}$ transforms trivially under permutations and is therefore in $\ch^{(3)}_{\lambda_0}$. This is just the statement that the symmetric subspace  $\ch^{(3)}_{\lambda_0}$ of $\ch^{(3)} \simeq \ch^{(3),\spatial} \otimes \ch^{(3),\spin}$ contains the tensor product of the symmetric subspaces of the factors, namely $\ch^{(3),\spatial}_{\lambda_0} \otimes \ch^{(3),\spin}_{\lambda_0} \subset \ch^{(3)}_{\lambda_0}$. However, as is well-known, the symmetric subspace of a composite system also contains the tensor product of the antisymmetric subspaces of the factors: $\ch^{(3),\spatial}_{\lambda_1} \otimes \ch^{(3),\spin}_{\lambda_1} \subset \ch^{(3)}_{\lambda_0}$ where $\lambda_1$ labels the anti-symmetric subspace. Similarly for $\lambda_2$ the $(2,1)$ representation of $S_3$ there will be elements in $\ch^{(3),\spatial}_{\lambda_2} \otimes \ch^{(3),\spin}_{\lambda_2} $ which are invariant under $U^{(3),\spatial}_{\lambda_2}(\pi) \otimes U^{(3),\spin}_{\lambda_2}(\pi)$, i.e. which lie in $\ch^{(3)}_{\lambda_0}$

For  $\lambda_2 = (2,1)$ we  define  a permutation-invariant state in $ V_{\lambda_2 }^\spatial  \otimes V_{\lambda_2 }^\spin $, where $V_{\lambda_2 }^\spatial \simeq  V_{\lambda_2 }^\spin \simeq \Cl^2$ both transform under the $(2,1)$ irrep of $S_3$. From the well known fact that the singlet state is invariant under $U \otimes U$ we can define the following Bosonic state in $\ch^{(3)}_{\lambda_0} \subset \ch^{(3)}$:
\begin{align}\label{eq: inv_pur_21}
    \ket{\Psi} = \left(\frac{1}{\sqrt{2}}( \ket{e_0}_{V_{\lambda_2}^\spatial} \otimes \ket{e_1}_{V_{\lambda_2}^\spin} - \ket{e_1}_{V_{\lambda_2}^\spatial} \otimes \ket{e_0}_{V_{\lambda_2}^\spin})\right) \otimes \ket{\psi}_{W_{\lambda_2}^\spatial,W_{\lambda_2}^\spin},
\end{align}
where $\{\ket{e_i}_{V_{\lambda_2}^\spatial}\}_{i = 0,1}$  is a basis for $V_{\lambda_2}^\spatial$ (and similarly for $V_{\lambda_2}^\spin$) and $\ket{\psi}_{W_{\lambda_2}^\spatial,W_{\lambda_2}^\spin}$ an arbitrary state in $W_{\lambda_2}^\spatial,W_{\lambda_2}^\spin$ . This is clearly invariant since a simultaneous permutation $U^{(3)}(\pi) = U^{(3),\spatial}(\pi) \otimes U^{(3),\spin}(\pi)$ of the spatial and spin indices acts on $\ch_{\lambda_2}^{(3),\spatial} \otimes \ch_{\lambda_2}^{(3),\spin}$ as $U_{V_{\lambda_2}^\spatial}(\pi) \otimes  U_{V_{\lambda_2}^\spin}(\pi) \otimes \I_{W_{\lambda_2}^\spatial}  \otimes \I_{W_{\lambda_2}^\spin} $ and the singlet state over the first two factors is invariant under any $U \otimes U$.

Let us ignore the isospin degrees of freedom and trace out $V_{\lambda}^\spin \otimes W_{\lambda}^\spin$:
\begin{align}\label{eq:para_21}
   \rho_{\rm para}:={\rm Tr}_{V_{\lambda_2}^{\rm spin}W_{\lambda_2}^{\rm spin}}|\Psi\rangle\langle\Psi|=\frac 1 2 \mathbf{1}_{V_{\lambda_2}^{\rm spatial}} \otimes \rho_{W_{\lambda_2}^\spatial},
\end{align}
where $\rho_{W_{\lambda_2}^\spatial} =\Tr_{W_{\lambda_2^\spin}}(\ketbra{\psi}{\psi}_{W_{\lambda_2}^\spatial,W_{\lambda_2}^\spin}) $. This is indeed a parafermionic state of three particles with just spatial degrees of freedom, since $\sum_i  \ketbra{e_i}{e_i}_{V_{\lambda}^\spatial}$ is the identity on the space $V_{\lambda}^\spatial$ transforming under the $(2,1)$ irrep of $S_3$ and $V_{\lambda}^\spatial$ transforms trivially under $S_3$.
\end{example}

\pagebreak
\section*{Supplementary Note 3 -- Complete invariance}

\subsection{Failure of complete invariance for emergent paraparticles}
\label{SubsecFailure}

The effective paraparticles in Example~\ref{ExEmergent} above are permutation-invariant. But does it make sense to require complete invariance? The answer is in the negative, as we now show.

Given an agent Alice using the effective paraparticles of~\Cref{eq:para_21}, an adversary Eve without restriction would have access to the purification given in~\Cref{eq: inv_pur_21}. 

Alice can implement permutations only on the spatial components, and not the spin (since by construction of the effective paraparticles she does not have access to them). Her state $ \rho_{\rm para}$ is indeed invariant under $U_{V_{\lambda_2}^\spatial}(\sigma)$ for all $\sigma \in S_3$. 

The purification $\ket{\Psi}$ which Eve has access to is not invariant under $U_{V_{\lambda_2}^\spatial}(\sigma) \otimes \I_{W_{\lambda_2}^\spatial} \otimes  \I_{V_{\lambda_2}^\spin}(\sigma) \otimes \I_{W_{\lambda_2}^\spin} $ and therefore does not obey complete invariance. Remember that the requirement of being a Bosonic state meant invariance under a different condition, namely $U_{V_{\lambda_2}^\spatial}(\sigma) \otimes \I_{W_{\lambda_2}^\spatial} \otimes  U_{V_{\lambda_2}^\spin}(\sigma) \otimes \I_{W_{\lambda_2}^\spin} $.

We can see that the purification is not invariant under $U_{V_{\lambda_2}^\spatial}(\sigma)$ from the following facts about the $(2,1)$ irrep of $S_3$, namely that for $\sigma = (123)$:
\begin{align}
    U(\sigma) \ket{e_0} = \omega \ket{e_0},\\
    U(\sigma) \ket{e_1} = \bar \omega \ket{e_1},
\end{align}
where $\omega = e^{\frac{2 \pi i}{3}}$.
Therefore:
\begin{align}
    &U_{V_{\lambda_2}^\spatial}(\sigma) \otimes \I_{W_{\lambda_2}^\spatial} \otimes  \I_{V_{\lambda_2}^\spin}(\sigma) \otimes \I_{W_{\lambda_2}^\spin}  \ket{\Psi} \nonumber \\
    =&\frac{1}{\sqrt{2}}\left( \omega \ket{e_0}_{V_{\lambda_2}^\spatial} \otimes \ket{e_1}_{V_{\lambda_2}^\spin} - \bar \omega \ket{e_1}_{V_{\lambda_2}^\spatial} \otimes \ket{e_0}_{V_{\lambda_2}^\spin}\right) \otimes \ket{\psi}_{W_{\lambda_2}^\spatial,W_{\lambda_2^\spin}}\\ \nonumber
    \simeq& \frac{1}{\sqrt{2}}\left(  \ket{e_0}_{V_{\lambda_2}^\spatial} \otimes \ket{e_1}_{V_{\lambda_2}^\spin} - \omega \ket{e_1}_{V_{\lambda_2}^\spatial} \otimes \ket{e_0}_{V_{\lambda_2}^\spin}\right) \otimes \ket{\psi}_{W_{\lambda_2}^\spatial,W_{\lambda_2^\spin}} \nonumber \\ \not\simeq &\ket{\Psi}
\end{align}
where $\simeq$ denotes equivalence up to global phase.

The above example shows that given an effective system of paraparticles emerging from an underlying Bosonic system, the assumption of complete invariance fails exactly because the full underlying Bosonic state is not invariant under permutations acting solely on the parastatistical subsystem.

\subsection{Complete invariance for compact groups}
\label{SuppSecCompleteInv}

In this section, we show that our results on complete invariance do not only hold for the permutation group, but more generally for unitary representations of compact groups $G$ on separable Hilbert spaces. We will always assume that $G$ is, as a topological space, a Hausdorff space.

Since Schur's Lemma applies to unitary representations of arbitrary locally compact Hausdorff groups~\cite[Proposition 5.8]{murnaghan}, and since finite-dimensional representations of such groups are completely reducible, \Cref{supp:lem:LemSchurConsequence} and therefore all the results in this section also apply to \textit{finite-dimensional} representations of such more general groups.

\begin{definition}[Extension of a state]
    An extension of a state $\rho_S$ is a state $\rho_{SA}$ such that $\Tr_A(\rho_{SA}) = \rho_S$.
\end{definition}

\begin{definition}[Invariant  (weakly symmetric) state]
    An invariant (or weakly symmetric) state $\rho_S$ for a given symmetry $U(g)$ is such that $\rho_S = U(g) \rho_S U^\dagger(g)$ for all $g \in G$.
\end{definition}

By~\Cref{lem:inv_operator_tensor} (where $\ch_A$ is trivial) it follows that a weakly invariant state is of the form
\begin{align}
    \rho = \bigoplus p_\lambda \frac{\I_{\cm_\lambda}}{d_\lambda} \otimes \rho_{\cn_\lambda}.
\end{align}

\begin{definition}[Strongly symmetric state]
    A strongly symmetric state $\rho$ for a symmetry $U(g) \simeq \bigoplus_\lambda U_\lambda(g) \otimes \I_{\cn_\lambda}$ has full support on a one-dimensional irreducible representation of $G$:
    \begin{align}
        \rho = \ketbra{e_\lambda}{e_\lambda}_{\cm_\lambda} \otimes \rho_{\cn_\lambda},
    \end{align}
    where $\cm_\lambda$ carries a one-dimensional representation of $G$: $U_\lambda(g) = e^{i \theta_\lambda(g)}$.
\end{definition}

\begin{definition}[Invariant extension]
    Given a state $\rho_S$ an extension $\rho_{SA}$ is invariant iff 
    \begin{align}\label{eq:invariant_extension}
        (U_S(g) \otimes \I_A) \rho_{SA} (U_S^\dagger(g) \otimes \I_A) = \rho_{SA}.
    \end{align}
\end{definition}

It is immediate that the existence of an invariant extension $\rho_{SA}$ for a state $\rho_S$ implies that $\rho_S$ is invariant: 
\begin{align}
     &(U_S(g) \otimes \I_A) \rho_{SA} (U_S^\dagger(g) \otimes \I_A) = \rho_{SA}\\ &\implies U_S(g) \rho_S U_S^\dagger(g) = \rho_S.
\end{align}

\begin{lemma}[Invariant purification]\label{lem:invariant_purification}
    If a state $\rho_S = U_S(g) \rho_S U^\dagger_S(g)$ has an invariant purification $\ket{\psi}_{SA}$ then $\rho_S$ has support on a single one-dimensional irreducible  representation $\lambda$ of $G$:
    \begin{align}
        \rho_S = \ketbra{e_\lambda}{e_\lambda}_{\cm_\lambda} \otimes \rho_{\cn_\lambda^S}.
    \end{align}
\end{lemma}
\begin{proof}
    A generic invariant state $\rho_S$ has the form
    \begin{align}
        \rho_S = \bigoplus_\lambda p_\lambda \frac{\I_{\cm_\lambda^S}}{d_\lambda} \otimes \rho_{\cn_\lambda^S}.
    \end{align}
        The invariant operators on $\ch_S \otimes \ch_A$ under $U_S(g) \otimes \I_A$ are of the form
    \begin{align}
     \bigoplus_\lambda \I_{\cm_\lambda^S}  \otimes A_{\cn_\lambda^S, \ch_A}.
    \end{align}

    A generic invariant extension $\rho_{SA}$ of $\rho_S$ therefore has the form
    \begin{align}
         \rho_{SA} = \bigoplus_\lambda p_\lambda \frac{\I_{\cm_\lambda^S}}{d_{\cm_\lambda^S}} \otimes \rho_{\cn_\lambda^S, \ch_A} ,
    \end{align}
    where
    \begin{align}
        \Tr_{\ch_A}( \rho_{\cn_\lambda^S, \ch_A}) =  \rho_{\cn_\lambda^S}.
    \end{align}
    Purity of $\rho_{SA} $ implies that
    \begin{align}
        \rho_{SA} =  \ketbra{e_\lambda}{e_\lambda}_{\cm_\lambda^S} \otimes \ketbra{\psi_\lambda}{\psi_\lambda}_{\cn_\lambda^S, \ch_A},
    \end{align}
    where $\lambda$ is a one-dimensional representation: $ \frac{\I_{\cm_\lambda^S}}{d_{\cm_\lambda^S}} = \ketbra{e_\lambda}{e_\lambda}_{\cm_\lambda^S} $.

    This implies that $\rho_S = \ketbra{e_\lambda}{e_\lambda}_{\cm_\lambda} \otimes \rho_{\cn_\lambda^S}$, where $\lambda$ is a one-dimensional irreducible representation.
    \end{proof}

\begin{lemma}\label{lem:completeinvgeneral}
The following conditions are all equivalent, and can be used to define what it means that a quantum state $\rho_S$ is \textbf{completely invariant} under  a representation $U(g)$, $g \in G$:
\begin{itemize}
    \item[(i)] \Cref{eq:invariant_extension} holds for some purification $\rho_{SA}$ of $\rho_S$;
    \item[(ii)] \Cref{eq:invariant_extension} holds for all purifications $\rho_{SA}$ of $\rho_S$;
    \item[(iii)] \Cref{eq:invariant_extension} holds for all extensions $\rho_{SA}$ of $\rho_S$.
\end{itemize}
Moreover, all three are equivalent to 
\begin{itemize}
    \item[(iv)]
    $\rho_S$ has full support on a single one-dimensional irreducible representation $\lambda$ of $G$.  
\end{itemize}
\end{lemma}

\begin{proof}
    Clearly, $(iii)\Rightarrow(ii)\Rightarrow(i)$, and by \Cref{lem:invariant_purification} $(i)$  implies (iv), i.e. that  $\rho_S = \ketbra{e_\lambda}{e_\lambda}_{\cm_\lambda} \otimes \rho_{\cn_\lambda^S}$, where $\lambda$ is a one-dimensional irreducible representation.

  Let us finally show that (iv) $\Rightarrow$ (iii). An arbitrary extension $\rho_{SA}$ of $\rho_S$ is therefore of the form
  \begin{align}
      \rho_{SA} = \ketbra{e_\lambda}{e_\lambda}_{\cm_\lambda} \otimes \rho_{\cn_\lambda^S, \ch_A},
  \end{align}
  where $\Tr_{\ch_A}(\rho_{\cn_\lambda^S, \ch_A}) = \rho_{\cn_\lambda^S}$. The state $\rho_{SA}$ is invariant under $U_S(g) \otimes \I_A$ by construction, since the action of $U_S(g) \otimes \I_A$ on $\cm_\lambda^S \otimes \cn_\lambda^S \otimes \ch_A$ is
  \begin{align}
      U_\lambda(g) \otimes \I_{\cn_\lambda^S } \otimes \I_{\ch_A } = e^{i \theta(g)} \ketbra{e_\lambda}{e_\lambda}_{\cm_\lambda} \otimes \I_{\cn_\lambda^S } \otimes \I_{\ch_A },
  \end{align}
  since $\cm_\lambda^S$ is one-dimensional.
\end{proof}

This implies the following theorem:
\begin{theorem}
\label{TheoremCompleteInvariancegeneral}
A quantum state $\rho_S$ is completely invariant under some compact group $G$ with representation $U_S(g)$, $g\in G$, if and only if it is fully supported on a subspace transforming under a one-dimensional representation of $G$.
\end{theorem}

Note that an invariant state $\rho_S$ with support on multiple irreducible subspaces, including subspaces with irreducible representations of dimension strictly greater than 1, may have \textit{some} extensions $\rho_{SA}$ which are invariant under the conjugate action of $U_S(g) \otimes \I_A$. But \Cref{lem:completeinvgeneral} tells us that this extension cannot be pure, and moreover that this cannot be the case for all extensions. 

Since complete invariance entails that a state $\rho_S$ must have full support on a subspace transforming according to a one-dimensional representation of $G$, we may wonder if standard invariance of a state $\rho_S$ can be motivated by appealing to an environment also. The following theorem shows that this is indeed the case.

\begin{theorem}\label{thm:invariance_purification}
A quantum state $\rho_S$ is (weakly) invariant under some compact group $G$, i.e.\ $\rho_S = U(g) \rho_S U_S^\dagger(g)$, if and only if there exists a purification $\ket{\psi}_{SA}$ of $\rho_S$ on some ancilla $\ch_A$ with some representation $U_A(g)$ such that $\big(U_S(g) \otimes U_A(g)\big) \ketbra{\psi}{\psi}_{SA} \big(U_S(g) \otimes U_A(g)\big)^\dagger = \ketbra{\psi}{\psi}_{SA}$. 
\end{theorem}

\begin{proof}
    We first prove the $\Leftarrow$ direction. Consider a purification  $\ket{\psi}_{SA}$ of $\rho_S$ such that $\big(U_S(g) \otimes U_A(g)\big) \ketbra{\psi}{\psi}_{SA} \big(U_S(g) \otimes U_A(g)\big)^\dagger = \ketbra{\psi}{\psi}_{SA}$. Then:
    \begin{align}
        \rho_S &= \Tr_A (\ketbra{\psi}{\psi}_{SA}) \nonumber \\&= \Tr_A (\big(U_S(g) \otimes U_A(g)\big) \ketbra{\psi}{\psi}_{SA} \big(U_S(g) \otimes U_A(g)\big)^\dagger) \nonumber \\
        &= U_S(g) \Tr_A (\ketbra{\psi}{\psi}_{SA}) U_S^\dagger(g) = U_S(g)\rho_S U_S^\dagger(g).
    \end{align}

        For the $\Rightarrow$ direction we assume there exists a (weakly) invariant state $\rho_S =U(g) \rho_S U_S^\dagger(g)$. Then by~\Cref{supp:lem:LemSchurConsequence},
        \begin{align}
            \rho_S = \bigoplus_\lambda p_\lambda \frac{\I_{\cm_\lambda^S}}{d_\lambda} \otimes \rho_{\cn_\lambda^S}\,\,.
        \end{align}
        Let $\ch_A \simeq \bar \ch_S$ carry the conjugate representation $U_A(g) \simeq \bar U_S(g)$. Then 
        \begin{align}
            \ch_A \simeq \bar \ch_S \simeq \bigoplus_\lambda \bar \cn_\lambda^A \otimes \bar \cm_\lambda^A,
        \end{align}
        where $\cn_\lambda^A$ carries the representation $\bar U_\lambda$.
        
        Now define the state
        \begin{align}
            \ket{\psi}_{SA} = \bigoplus_\lambda  \sqrt{p_\lambda} \sum_{i_\lambda} \ket{i_\lambda}_{\cm_\lambda^S} \ket{i_\lambda}_{\bar \cm_\lambda^A} \otimes \ket{\psi^\lambda}_{\cn_\lambda^S, \bar \cn_\lambda^A}
        \end{align}
        with $\ket{\psi^\lambda}_{\cn_\lambda^S, \bar \cm_\lambda^A}$ a purification of $\rho_{\cn_\lambda^S}$, which always exists since $\dim(\cn_\lambda^S) = \dim(\bar \cn_\lambda^A)$. 

        This is a purification of $\rho_S$ which is invariant under $U_S(g) \otimes U_A(g)$ by construction.
\end{proof}

\section*{Supplementary Note 4 -- Structural details on quantum permutations}

\subsection{Applicability of our formalism in second quantization: particle permutations}
\label{supp:sec:StandPermSecQuant}

While the action of standard permutations in first quantization is given by 
\begin{equation}
\label{methods:eq:PlacePermAction}
U(\pi)|x_1,\ldots,x_N\rangle=|x_{\pi^{-1}(1)},\ldots,x_{\pi^{-1}(N)}\rangle,
\end{equation}
usually referred to as a \textit{label} or \textit{place} permutation, it can generally not be defined consistently this way on Fock states
\begin{align}
\label{methods:eq:FockState}
 |k_1,...,k_N\rangle_\mathcal{F} := a_{k_1}^\dagger...a_{k_N}^\dagger|0\rangle_\mathcal{F}
\end{align}
and the associated creation operators in second quantization. To see this, consider the parastatistical model from Example \ref{supp:ex:VolkovPara}: given the state $\left(a^\dagger_{2}a^\dagger_{1}a^\dagger_{3}+  a^\dagger_{2}a^\dagger_{3}a^\dagger_{1} \right)|0\rangle_\mathcal{F}$ then exchanging the positions of the first two creation operators gives $\left(a^\dagger_{1}a^\dagger_{2}a^\dagger_{3}+  a^\dagger_{3}a^\dagger_{2}a^\dagger_{1} \right)|0\rangle_\mathcal{F} = 0$, entailing that place permutations are not unitary.  Instead, standard permutations in second quantization, which are referred to as \textit{particle} permutations, act by permuting the \textit{modes}. For example, the particle permutation corresponding to exchanging modes 1 and 2 from above can be represented by a unitary operator $\bar U((12))$ that effectively exchanges the creation operators $a_1^\dagger$ and $a_2^\dagger$ and transforms $\left(a^\dagger_{2}a^\dagger_{1}a^\dagger_{3}+  a^\dagger_{2}a^\dagger_{3}a^\dagger_{1} \right)|0\rangle_\mathcal{F}$ into $\left(a^\dagger_{1}a^\dagger_{2}a^\dagger_{3}+  a^\dagger_{1}a^\dagger_{3}a^\dagger_{2} \right)|0\rangle_\mathcal{F} \neq 0$.

Following Stolt and Taylor~\cite{Taylor3} we define particle permutations by considering tuples $K=(k_1,\ldots,k_N)$ of modes, say momenta, that are pairwise distinct, $k_i\neq k_j$ for $i\neq j$, and we choose lexicographical order, assuming that $k_1<k_2<\ldots<k_N$. This is simply a convention to ensure we do not consider tuples $K\neq K'$ that have the same set of entries, but in different order. For our purposes, we are only concerned with the case of a single particle occupation per mode, as this suffices to derive all our results. We thus refer to \cite{Taylor3} for the construction with different occupation numbers of the modes. 

Now, for every $K$, we have a subspace $\mathcal{H}_K$ that is spanned by the state vectors $|k_{\pi(1)},\ldots,k_{\pi(N)}\rangle_{\mathcal{F}}$, for all permutations $\pi\in S_N$. This subspace carries a representation $\pi\mapsto \bar U_K(\pi)$ of $S_N$, which transforms creation operators $a_{k_i}^\dagger$ into $\bar U_K(\pi) a_{k_i}^\dagger\bar U_K(\pi)^\dagger=a_{k_{\pi(i)}}^\dagger$. The space $\mathcal{H}_K$ is of dimension less than or equal to $N!$, which in the case of Bosons or Fermions is one-dimensional.

The $\bar U_K(\pi)$ define reducible representations of $S_N$ and we get a decomposition
\begin{align}
\mathcal{H}_K=\bigoplus_\lambda \mathcal{H}_{\lambda,K},
\end{align}
where $\mathcal{H}_{\lambda,K}$ carries the $\lambda$-irrep of $S_N$, and the sum is over all admissible Young diagrams. That is, as shown in~\cite{Ohnuki}, each irrep $\lambda$ corresponding to an admissible Young diagram appears exactly once in this decomposition.

We can now consider an external system described by some Hilbert space $\mathcal{N}$, write
\begin{equation}
   \mathcal{H}_K\otimes \mathcal{N}=\bigoplus_\lambda \mathcal{H}_{\lambda,K}\otimes \mathcal{N}.
   \label{eqMultSpace}
\end{equation}
Alternatively, the particles could carry some internal degree of freedom, such as spin, say with total spin $s$. In this case, we would expect that the creation operators $a_k^\dagger$ are replaced by creation operators $a_{k,\sigma}^\dagger$, where $\sigma$ can take several different values (e.g.\ $\sigma\in\{-s,-s+1,\ldots,s\}$). Consider the special case that all particles are prepared in the same spin $\sigma$, as is the case in the experimental proposal by Roos et al.~\cite{Roos}. Then, for every fixed $\sigma$, the operators $\{a^\dagger_{k,\sigma}\}_k$ satisfy relations~(\ref{supp:eq:TrilRel}). Since states for different $\sigma$ are orthogonal, we obtain a particle Hilbert space that is also of the form~(\ref{eqMultSpace}), where $\mathcal{N}\simeq\mathbb{C}^{2s+1}$. More generally, for every choice of spins $\sigma_1,\sigma_2,\ldots,\sigma_N$, there will be a subspace spanned by vectors $|(k_{\pi(1)},\sigma_{\pi(1)}),\,\ldots,(k_{\pi(N)},\sigma_{\pi(N)}\rangle_{\mathcal{F}}$ which individually behaves like the paraparticle subspace without internal degrees of freedom described above. This also gives us a Hilbert space of the form~(\ref{eqMultSpace}), where now $\mathcal{N}\simeq \mathbb{C}^{(2s+1)^N}$. In all cases, we obtain a multiplicity space as in Eq.~(\ref{eqMultSpace16}) and Eq.~(\ref{eqMultSpace17}) that allows us to construct quantum permutations by conditioning on observables supported on it.

\subsection{More general quantum permutation groups}
\label{supp:sec:moreGenQP}

In principle, we could have started from a more general version in the definition of quantum permutations: instead of conditioning only the argument $\pi_{\lambda,j}$ of the representation map $\pi \mapsto e^{i\theta(\pi)}U(\pi)$ on $\lambda, j$, we could also condition the map itself on $\lambda,j$ as in Eq.~(\ref{methods:eq:MoreGeneral}) in the main text: $\pi \mapsto W_{\lambda,j}(\pi):=e^{i\theta_{\lambda,j}(\pi)}U(\pi)$. To do so, we will require that the \textit{unconditional} special case should still reproduce the standard permutations: applying the same permutation $\pi$ in every branch must projectively act like $U(\pi)$; that is, $V(\pi,\pi,\ldots,\pi)\rho V(\pi,\pi,\ldots,\pi)^\dagger = U(\pi)\rho U(\pi)^\dagger$.

As in the proof of Lemma \ref{methods:lem:proj_SNQ_rep} in the main text, $V(\vec{\sigma})V(\vec{\pi})=\omega(\vec{\sigma},\vec{\pi})V(\vec{\sigma} \vec{\pi})$ then implies
\begin{align}
\omega(\vec{\sigma},\vec{\pi})=e^{i\theta_{\lambda,j}(\sigma_{\lambda,j})}e^{i\theta_{\lambda,j}(\pi_{\lambda,j})}e^{-i\theta_{\lambda,j}(\sigma_{\lambda,j}\pi_{\lambda,j})}
\end{align}
for all $\lambda, j$. Assuming $Q>1$ without loss of generality, suppose that $\vec{\sigma}'$ and $\vec{\pi}'$ are elements of $S_N^Q$ that agree with $\vec{\sigma}$ and $\vec{\pi}$ on at least one entry $(\lambda,j)$, i.e.\ there exists some $(\lambda_0,j_0)$ such that $\sigma_{\lambda_0,j_0}=\sigma'_{\lambda_0,j_0}$ and $\pi_{\lambda_0,j_0}=\pi'_{\lambda_0,j_0}$. Then, we have
\begin{align}
\omega(\vec{\sigma},\vec{\pi})&=e^{i\theta_{\lambda_0,j_0}(\sigma_{\lambda_0,j_0})}e^{i\theta_{\lambda_0,j_0}(\pi_{\lambda_0,j_0})}e^{-i\theta_{\lambda_0,j_0}(\sigma_{\lambda_0,j_0}\pi_{\lambda_0,j_0})}\nonumber \\
&=e^{i\theta_{\lambda_0,j_0}(\sigma'_{\lambda_0,j_0})}e^{i\theta_{\lambda_0,j_0}(\pi'_{\lambda_0,j_0})}e^{-i\theta_{\lambda_0,j_0}(\sigma'_{\lambda_0,j_0}\pi'_{\lambda_0,j_0})} \nonumber \\
&=\omega(\vec{\sigma}',\vec{\pi}').
\end{align}
Let $\vec{\sigma}'',\vec{\pi}''$ be arbitrary elements of $S_N^Q$. Then we can always find some pair $\vec{\sigma}',\vec{\pi}'$ that agree with the pair $\vec{\sigma},\vec{\pi}$ in at least one entry, say $\lambda_0,j_0$, and that also agrees with the pair $\vec{\sigma}'',\vec{\pi}''$ in at least one entry, say $\lambda_1,j_1$. Thus,
\begin{align}
\omega(\vec{\sigma}',\vec{\pi}')&=e^{i\theta_{\lambda_1,j_1}(\sigma'_{\lambda_1,j_1})}e^{i\theta_{\lambda_1,j_1}(\pi'_{\lambda_1,j_1})}e^{-i\theta_{\lambda_1,j_1}(\sigma'_{\lambda_1,j_1}\pi_{\lambda_1,j_1})}\nonumber \\
&=e^{i\theta_{\lambda_1,j_1}(\sigma''_{\lambda_1,j_1})}e^{i\theta_{\lambda_1,j_1}(\pi''_{\lambda_1,j_1})}e^{-i\theta_{\lambda_1,j_1}(\sigma''_{\lambda_1,j_1}\pi''_{\lambda_1,j_1})} \nonumber \\
&=\omega(\vec{\sigma}'',\vec{\pi}''),
\end{align}
so $\omega(\vec{\sigma},\vec{\pi})=\omega(\vec{\sigma}'',\vec{\pi}'')=\omega$ is a constant that does not depend on $\vec{\sigma}$ or $\vec{\pi}$. The special case $\pi_{\lambda,j}=\I$ shows that $\omega=e^{i\theta}$ for $\theta:=\theta(\I)$, and hence $\vec{\pi}\mapsto e^{-i\theta}V(\vec{\pi})$ is a linear unitary representation. Furthermore, $\pi\mapsto e^{-i\theta}e^{i\theta_{\lambda,j}(\pi)}$ is a linear one-dimensional representation of $S_N$, i.e.\ either the trivial or the sign representation, depending on $\lambda,j$.

Now, even though this more general definition still qualifies as a projective representation of $S_N$, its action does not generally reduce to that of standard permutations  in the unconditional special case $\vec{\pi}=(\pi,\pi,\ldots,\pi)$. To see this, assume we have such a general map $\vec{\pi} \mapsto V(\vec{\pi})$. Then, setting $e^{i\theta}=1$, there exist pairs $\lambda_1,j_1$ and $\lambda_2,j_2$, where at least $\lambda_1 \neq \lambda_2$ or $j_1 \neq j_2$, such that $W_{\lambda_1,j_1}(\pi_{\lambda_1,j_1})=U(\pi_{\lambda_1,j_1})$ and $W_{\lambda_2,j_2}(\pi_{\lambda_2,j_2})=\sgn(\pi_{\lambda_2,j_2})U(\pi_{\lambda_2,j_2})$. For any $\pi \in S_N$, set $\vec{\pi}:=(\pi,...,\pi)$ and define
\begin{align}
    A(\pi):=U(\pi)^\dagger V(\vec{\pi})=V(\vec{\pi})U(\pi)^\dagger =\sum_{\lambda,j}s_{\lambda,j}(\pi)\I_{\lambda}\otimes P_{\lambda,j}, 
\end{align}
where $s_{\lambda,j}(\pi)\in \{1,\sgn(\pi)\}$. Furthermore, note that, since $P_{\lambda,j}$ is a rank-1 projector, for every $\lambda,j$ there exists an element $|\psi_{\lambda,j}\rangle$ such that $\I_\lambda\otimes P_{\lambda,j}|\psi_{\lambda,j}\rangle=|\psi_{\lambda,j}\rangle$. This implies
\begin{align}
    A(\pi)(|\psi_{\lambda_1,j_1}\rangle+|\psi_{\lambda_2,j_2}\rangle)=|\psi_{\lambda_1,j_1}\rangle+\sgn(\pi)|\psi_{\lambda_2,j_2}\rangle, 
\end{align}
showing that $A(\pi)\neq e^{i\theta(\pi)}\I$ for any $\theta(\pi)$, as it causes relative phases, and hence $V(\vec{\pi}) \neq e^{i\theta(\pi)}U(\pi)$. Therefore, $V(\vec{\pi})$ does not  projectively act like $U(\pi)$.

\subsection{Proof of Theorem~\ref{methods:th:MainTheorem2} in the main text}
\label{supp:sec:proofTh3}

Assuming the Bosonic case $U(\pi)$ (the Fermionic case $\sgn(\pi)U(\pi)$ works in the same way), recall the form of the quantum permutations~of equation (\ref{methods:eq:QPermutation1}) in the main text. They are given with respect to a fixed basis for each multiplicity space $\mathcal{N}_\lambda$. In order to get rid of this, we define 
\begin{align}
\label{eqU_R_definition}
U^{(\mathbf{R}
)}(\vec{\pi}):= \left(\bigoplus_\mu \I_{\mathcal{M}_\lambda} \otimes R_\mu\right)U(\vec{\pi})\left(\bigoplus_\nu \I_{\mathcal{M}_\lambda} \otimes R^\dagger_\nu\right)=\sum_{\lambda,j} U_\lambda(\pi_{\lambda,j}) \otimes R_\lambda P_{\lambda,j}R^\dagger_\lambda
\end{align}
for arbitrary unitaries $R_\lambda\in\text{U}(n_\lambda)$ and any $\lambda$. With this definition we have $U(\vec{\pi})=U^{(\I)}(\vec{\pi})$, i.e. $R_\lambda=\I_{\mathcal{N}_\lambda}$.

Now, we have seen that the requirement $[U^{(\I)}(\vec{\pi})_S\otimes\I_A,\rho_{SA}]=0$ for all $\vec{\pi}\in S_N^Q$ implies
\begin{align}
\label{eqRho_SA_form}
\rho_{SA}=p_{\rm Bos}\, \rho_{{\rm Bos},A}\oplus\bigoplus_{j,\lambda\neq\lambda_{\rm Bos}} p_{\lambda,j} \frac{\I_{\lambda,j}}{d_\lambda} \otimes \rho_A^{(\lambda,j)}
=p_{\rm Bos}\, \rho_{{\rm Bos},A}\oplus\bigoplus_{\lambda\neq\lambda_{\rm Bos}} \frac{\I_{\mathcal{M}_\lambda}}{d_\lambda} \otimes \sum_j p_{\lambda,j}P_{\lambda,j} \otimes \rho_A^{(\lambda,j)},
\end{align}
where $d_{\lambda}$ is the dimension of $\mathcal{H}_{\lambda,j}$, and the state of system $\mathcal{N}_\lambda$ is diagonal in the basis associated to the projectors $P_{\lambda,j}$. Demanding invariance with respect to all possible quantum permutations without a fixed basis means that we need to impose the stronger condition $[U^{(\mathbf{R})}(\vec{\pi})_S\otimes\I_A,\rho_{SA}]=0$ for all $\vec{\pi}\in S_N^Q$ and all $R_\lambda$. Therefore, we require
\begin{align}
(U^{(\mathbf{R})}(\vec{\pi})_S\otimes\I_A)\rho_{SA}(U^{(\mathbf{R})}(\vec{\pi})_S\otimes\I_A)^\dagger=\rho_{SA},
\end{align}
which, due to Eq.~(\ref{eqU_R_definition}), is equivalent to
\begin{align}
[U^{(\I)}(\vec{\pi})_S\otimes\I_A, \underbrace{(\oplus_\mu \I_{\mathcal{M}_\mu} \otimes R_\mu\otimes \I_A)\rho_{SA}(\oplus_\nu \I_{\mathcal{M}_\nu} \otimes R^\dagger_\nu\otimes \I_A}_{=:\rho^{(\mathbf{R})}_{SA}}]=0,
\end{align}
so $\rho^{(\mathbf{R})}_{SA}$ must be of the same form as $\rho_{SA}$ in Eq.~(\ref{eqRho_SA_form}) for all $R_\lambda$. We obtain
\begin{align}
\rho^{(\mathbf{R})}_{SA}=p_{\rm Bos}\,(R_{\rm Bos}\otimes \I_A) \rho_{{\rm Bos},A}(R^\dagger_{\rm Bos}\otimes \I_A)\oplus\bigoplus_{\lambda\neq\lambda_{\rm Bos}} \frac{\I_{\mathcal{M}_\lambda}}{d_\lambda} \otimes \sum_j p_{\lambda,j}R_\lambda P_{\lambda,j}R_\lambda^\dagger \otimes \rho_A^{(\lambda,j)},
\end{align}
and see that this immediately implies
\begin{align}
R_\lambda\left(\sum_j p_{\lambda,j}P_{\lambda,j}\right)R_\lambda^\dagger=\sum_j p'_{\lambda,j}P_{\lambda,j}
\end{align}
for some $p'_{\lambda,j}$ after tracing out the ancillary system $A$. This means that $R_\lambda(\sum_j p_{\lambda,j}P_{\lambda,j})R_\lambda^\dagger$ must be diagonal with respect to the $P_{\lambda,j}$-basis for all unitary conjugations by $R_\lambda$, which is only the case when $p_{\lambda,j}$ does not depend on $j$ for any $\lambda\neq \lambda_{\rm Bos}$. Due to $\text{dim}(\mathcal{N}_\lambda)=n_\lambda$, we have $p_{\lambda,j}=p_\lambda/n_\lambda$, such that $p_{\rm Bos}+\sum_{\lambda\neq \rm Bos}p_\lambda=1$. Note that this is only possible if $n_\lambda<\infty$. This already proves that in the case that at least one of the $N_\lambda$ for $\lambda\neq\lambda_{\rm Bos}$  is infinite-dimensional, there is no invariant state $\rho_{SA}$.

Let us continue with the case that all $n_\lambda<\infty$ for $\lambda\neq\lambda_{\rm Bos}$. With this, let us also make the claim that basis-independent invariance implies that for any $\lambda \neq \lambda_{\rm Bos}$, the states $\rho_A^{(\lambda,j)}$ cannot depend on $j$ either. To show this, we first note that, since $P_{\lambda,j}$ are rank-1 projectors of a projective measurement, we can find an orthogonal basis $\{|j\rangle_\lambda\}_j$ of $\mathcal{N}_\lambda$ such that $P_{\lambda,j}=|j\rangle\langle j|_\lambda$. Then, given unitaries $R_\lambda$ with $|j^{(\mathbf{R})}\rangle_\lambda:=R_\lambda|j\rangle_\lambda$, we have
\begin{align}
\rho^{(\mathbf{R})}_{SA}=p_{\rm Bos}\,(R_{\rm Bos}\otimes \I_A) \rho_{{\rm Bos},A}(R^\dagger_{\rm Bos}\otimes \I_A)\oplus\bigoplus_{\lambda\neq\lambda_{\rm Bos}} p_\lambda\frac{\I_{\mathcal{M}_\lambda}}{d_\lambda n_\lambda} \otimes \sum_j |j^{(\mathbf{R})}\rangle \langle j^{(\mathbf{R})}|_\lambda \otimes \rho_A^{(\lambda,j)}.
\end{align}
Now, if $[U^{(\I)}(\vec{\pi})_S\otimes\I_A,\rho^{(\mathbf{R})}_{SA}]=0$ holds, then it must follow that $\rho^{(\mathbf{R})}_{SA}$ has the form of Eq.~(\ref{eqRho_SA_form}) when expressed in the $P_{\lambda,j}$-basis. Due to $|j^{(\mathbf{R})}\rangle_\lambda=\sum_k(R_\lambda)_{kj}|k\rangle_\lambda$, we find
\begin{align}
\rho^{(\mathbf{R})}_{SA}=p_{\rm Bos}\,(R_{\rm Bos}\otimes \I_A) \rho_{{\rm Bos},A}(R^\dagger_{\rm Bos}\otimes \I_A)\oplus\bigoplus_{\lambda\neq\lambda_{\rm Bos}} p_\lambda\frac{\I_{\mathcal{M}_\lambda}}{d_\lambda n_\lambda} \otimes \sum_{k,l} |k\rangle \langle l|_\lambda \otimes \sum_j (R_\lambda)_{kj}(R^\dagger_\lambda)_{jl} \rho_A^{(\lambda,j)},
\end{align}
and thus, for all $R_\lambda$ we obtain the condition
\begin{align}
\sum_j (R_\lambda)_{kj}(R^\dagger_\lambda)_{jl} \rho_A^{(\lambda,j)}=0,
\end{align}
whenever $k\neq l$. In particular, let us now choose $R_\lambda$ to be rotations of the planes spanned by the first and the $i$-th coordinate axes by an angle $\alpha$. Then we have $(R_\lambda)_{11}=(R_\lambda)_{ii}=\cos{\alpha}$, $(R_\lambda)_{1i}=-(R_\lambda)_{i1}=\sin{\alpha}$ and all other entries being equal to those of the identity. With this, the condition above implies $\rho_{A}^{(\lambda,1)}=\rho_{A}^{(\lambda,i)}$ for any $i$, and as a consequence, $\rho^{(\lambda,j)}_A$ must not depend on $j$ which shows the claim.

Finally, a state $\rho_{SA}$ which is invariant under all possible quantum permutations must have the form
\begin{align}
\rho_{SA}=p_{\rm Bos}\, \rho_{{\rm Bos},A}\oplus\bigoplus_{j,\lambda\neq\lambda_{\rm Bos}} p_{\lambda} \frac{\I_{\lambda,j}}{d_\lambda n_\lambda} \otimes \rho_A^{(\lambda)}
=p_{\rm Bos}\, \rho_{{\rm Bos},A}\oplus\bigoplus_{\lambda\neq\lambda_{\rm Bos}} p_{\lambda} \frac{\I_{\mathcal{M}_\lambda\otimes\mathcal{N}_\lambda}}{d_\lambda n_\lambda} \otimes \rho_A^{(\lambda)},
\end{align}
with non-bosonic post-measurement states
\begin{align}
\frac{\I_{\mathcal{M}_\lambda\otimes\mathcal{N}_\lambda}}{d_\lambda n_\lambda} \otimes \rho_A^{(\lambda)}.
\end{align}
Tracing out the ancillary system $A$, we obtain the completely mixed state for any $\lambda\neq \lambda_{\rm Bos}$.

\section*{Supplementary Note 5 -- Quantum Reference Frames}

\subsection{Application to Bosons and Fermions}
\label{supp:sec:QRFsApplBosFerm}

In this subsection, we provide a swift introduction to Quantum Reference Frames (QRFs), thereby motivating the form of QRF transformations as conditional transformations, and we show how these can be applied in the context of Bosons and Fermions. In particular, we will see that the QRF transformation picture corresponds to the implementation of some symmetry \textit{on the level of observables}: only observables invariant under the symmetry are deemed measurable, but all quantum states are allowed as descriptions of preparation procedures. This is the convention of invariance that we adopt for this subsection (see the Methods section in the main text).

Let us begin with a simple illustrative example of QRF transformations. Consider three particles $A, B, C$ in one dimension subject to translation-invariance, following the analysis of Giacomini et al.~\cite{Giacomini}. Suppose that the state in position representation is a product state,
\begin{align}
|\psi_{|A}\rangle=|0\rangle_A\otimes \frac 1 {\sqrt{2}} \left(\strut|1\rangle_B+|2\rangle_B\right)\otimes |3\rangle_C
\end{align}
(to prevent some mathematical difficulties, let us assume that all positions are integers, like on a lattice). This has sometimes been said to describe the quantum state ``as seen by $A$'', because $A$ is placed in the origin. Under this intuition, we would obtain the state ``as seen by $C$'' by translating everything by $-3$,
\begin{align}
|\psi_{|C}\rangle=|-3\rangle_A\otimes \frac 1 {\sqrt{2}} \left(\strut|-2\rangle_B+|-1\rangle_B\right)\otimes |0\rangle_C.
\end{align}
But what would the state ``as seen by $B$'' look like? No proper translation can move $B$ to the origin. The idea is now to translate in superposition: one of the branches of $\psi_{|A}$ is translated by $-1$, and the other by $-2$, and this is done coherently, obtaining
\begin{align}
|\psi_{|B}\rangle=|0\rangle_B\otimes\frac 1 {\sqrt{2}}\left(\strut |-1\rangle_A\otimes|2\rangle_C+|-2\rangle_A\otimes|1\rangle_C\right).
\end{align}
This is typically interpreted as having the particles $A$, $B$ and $C$ themselves as quantum reference frames, i.e.\ quantum analogues of classical rods or clocks, and notions of entanglement being frame-dependent. Here we take a different perspective described in~\cite{KrummHoehnMueller,HoehnKrummMueller}: there is an algebra of observables $X$ that can be measured under a certain condition of translation-invariance, and we have
\begin{align}
\langle\psi_{|A}|X|\psi_{|A}\rangle=\langle\psi_{|B}|X|\psi_{|B}\rangle=\langle\psi_{|C}|X|\psi_{|C}\rangle
\end{align}
for all of them. This means that the $\psi_{|S}$ are alternative, equally valid descriptions of one and the same quantum state. Choosing one representation over another can be interpreted as choosing a quantum coordinate system; in this case, a system that aligns the origin with one of the particles. A \textit{QRF transformation} is a unitary transformation that maps from one quantum coordinate system to another. The starting point are the global translations $T_d$, translating the global state by distance $d$, acting as
\begin{align}
V(T_d)\left(\strut|a\rangle_A|b\rangle_B|c\rangle_C\right)=|a+d\rangle_A|b+d\rangle_B|c+d\rangle_C.
\label{eqRepOfTranslation}
\end{align}
QRF transformations are coherently controlled versions of this. For example, mapping from the quantum coordinate system relative to $A$ to the one relative to $B$ is achieved by the map
\begin{align}
U_{A\to B}=\sum_{d\in\mathbb{Z}} V(T_{-d})P_d,
\end{align}
where $\{P_d\}_{d\in\mathbb{Z}}$ is a projective measurement of the distance $d$ between particles $B$ and $A$, i.e.
\begin{align}
P_d|a\rangle_A|b\rangle_B|c\rangle_C=\left\{\begin{array}{cl} |a\rangle_A|b\rangle_B|c\rangle_C & \mbox{if }b-a=d\\
0 & \mbox{otherwise}.
\end{array}\right.
\end{align}
It implements the transformation from QRF $A$ to $B$, i.e.\ $U_{A\to B}|\psi_{|A}\rangle=|\psi_{|B}\rangle.$ It is a global translation, conditioned on a translation-invariant quantity, namely on the relative distance between $A$ and $B$. QRF transformations in general are understood as coherently controlled classical reference frame changes~\cite{HametteGalley}, where the controlling branches are themselves defined in a frame-independent way~\cite{Cepollaro}, or as state-dependent gauge transformations~\cite{Butterfield}.

The states $\psi_{|S}$ relative to the different subsystems $S$ are all physically equivalent to the invariant state $\rho:=\int_{\mathcal{U}} U|\psi_{|S}\rangle\langle\psi_{|S}|U^\dagger\, dU$, where $\mathcal{U}$ is the group of QRF transformations; see~\cite{KrummHoehnMueller} for details. This construction is closely related to the \textit{perspective-neutral framework}, where an invariant pure state $|\psi_{\rm phys}\rangle$ on a subspace $\mathcal{H}_{\rm phys}$ with in general redefined inner product takes this role~\cite{HoehnUniverse2019,Vanrietvelde2020,Hamette2021,Vanrietvelde2023}.

The similarity with our definition of quantum permutations, Definition~\ref{main:def:QuantumPerm} in the main text, is clear: we apply a symmetry transformation conditioned on an invariant quantity. Hence, quantum permutations can be interpreted as QRF transformations, or, slightly more accurately, quantum coordinate transformations. They transform between different quantum conventions for labeling the particles. As a result of Theorem \ref{main:th:QuantPermInvRulesOutPara} in the main text, we find that there exist exactly two groups of quantum permutations, Bosonic and Fermionic ones, depending on the choice of projectively equivalent representations. Bosonic (Fermionic) QRF transformations will leave all Bosonic (Fermionic) observables invariant, thus providing groups of transformations between different but physically equivalent descriptions of Bosonic (Fermionic) quantum states.   

For example, for $N=2$ Fermions, the quantum permutations are given by $V_-(\vec{\pi}) = \sum_{\lambda, j}  \sgn(\pi_{\lambda,j}) U(\pi_{\lambda,j}) (\I_\lambda \otimes P_{\lambda,j})$. Hence the exchange of the two particles, denoted $(1,2)$, is represented as $V_-(\vec{\pi} = \overrightarrow{(1,2)})$, where $\overrightarrow{(1,2)}_{\lambda,j} = (1,2)$ for all $\lambda,j$, which simplifies to $V(\overrightarrow{(1,2)})= \sgn((1,2))U((1,2)) = -U(\rm swap)$. The measurable observables are all the operators $X$ that are fully supported on the antisymmetric subspace. Suppose we have a quantum state that we describe as $|\psi_{\rm asc}\rangle=|1,4\rangle$, where by convention we decided to describe the particles in ascending order. Then we could equivalently decide to describe them in descending order via $|\psi_{\rm des}\rangle:=V(\overrightarrow{(1,2)})|\psi_{\rm asc}\rangle=-|4,1\rangle$, and then
\begin{align}
\langle\psi_{\rm asc}|X|\psi_{\rm asc}\rangle=\langle\psi_{\rm des}|X|\psi_{\rm des}\rangle
\end{align}
for all measurable observables $X$. The states $\ket{\psi_{\rm asc}}$ and $\ket{\psi_{\rm des}}$ are therefore physically equivalent descriptions. These are not permutationally invariant, however the state $\ket{\psi_{\rm phys}} = \frac{1}{\sqrt{2}}(\ket{1,4} - \ket{4,1})$ is invariant (not just as a state, but also as a vector) under $V((1,2))$ and is physically equivalent to the states $\ket{\psi_{\rm asc}}$ and $\ket{\psi_{\rm des}}$. From this simple example, it is not obvious why the minus sign is significant in the definition of $V(\overrightarrow{(1,2)})$. Indeed, for unconditional permutations (i.e.\ $\pi_{\lambda,j} = \pi_{\lambda',j'}$), the choice of representation $-U(\rm swap)$ occurring in $V_-(\vec{\pi})$ is equivalent up to a global phase (hence also physically equivalent) to the choice $U(\rm swap)$. The choice of representation $-U(\rm swap)$ becomes significant for conditional permutations as becomes evident in a slightly more involved example:
\begin{example}
\label{ExTwoFermions}
Suppose that we have $N=2$ Fermions, such that the measurable observables are those supported on the antisymmetric subspace. Representing the swap of the two particles by $e^{i\theta}U(\rm swap)$ with some phase $\theta$, consider the operator $X=|\{2,3\}\rangle\langle\{1,4\}|_-$, where $|\{x_1,x_2\}\rangle_-:=\frac 1 {\sqrt{2}}\left(\strut|x_1,x_2\rangle-|x_2,x_1\rangle\right)$, and a state $|\psi\rangle=\frac 1 {\sqrt{2}}\left(\strut |2,3\rangle+|1,4\rangle\right)$. We apply the following quantum permutation: we swap the two particles if their distance is larger than $2$, and otherwise we do nothing. This maps the state $|\psi\rangle$ to $|\psi'\rangle=\frac 1 {\sqrt{2}}\left(\strut |2,3\rangle+e^{i\theta}|4,1\rangle\right)$, and
\begin{align}
\langle\psi|X|\psi\rangle=\langle\psi'|X|\psi'\rangle
\end{align}
if and only if $e^{i\theta}=-1$, explaining the choice of phase above. It follows from Theorem~\ref{main:th:QuantPermInvRulesOutPara} in the main text that this equation is true for all $X$ on the antisymmetric subspace. Hence, this quantum permutation maps between physically equivalent descriptions of the quantum state.
\end{example}
In two upcoming publications, we will further elaborate on this usage of the internal quantum reference frame  formalism for Bosons and Fermions: we will show how it clarifies the question of entanglement of indistinguishable particles in a particularly transparent way~\cite{upcoming}, and  that it is a simple example of a natural generalization of the perspective-neutral framework~\cite{Vanrietvelde2020,Hamette2021} where internal quantum reference frames need not correspond to subsystems~\cite{Upcoming2}.

\subsection{Implications for quantum covariance principles}
\label{supp:sec:ImplicationsQuantCovPrinciples}

An ambitious goal of the internal quantum reference frames (QRF) research program is to obtain physical predictions for some scenarios that would ultimately be described by quantum gravity. As an example, consider a thought experiment by de la Hamette et al.~\cite{QRFIndefiniteMetric}. We have a large mass (for the sake of the argument, planet Earth) in a superposition of  locations relative to some reference system, and a test particle moving in the resulting gravitational field. We cannot appeal to the known physical laws to  predict the movement of the test particle, since we expect to obtain some sort of superposition of gravitational fields and spacetime geometries which would fall into the regime of quantum gravity. However, the authors of~\cite{QRFIndefiniteMetric} suggest that it makes sense to postulate an ``extended symmetry principle'' described as the ``covariance of dynamical laws under quantum coordinate transformations''. Based on this principle, one may attempt to transform into a quantum coordinate system where Earth's location and hence the resulting gravitational field are definite, but the test particle is in a superposition state. Based on the generalized covariance principle, one would then solve the quantum equations of motion of the test particle in the  corresponding classical gravitational field and transform back, obtaining a description of the movement of the test particle in a superposition of gravitational fields. 

The results of our work show a subtle difficulty of such approaches: quantum coordinate transformations do not necessarily preserve all physical predictions, but only some of them. Before we proceed, let us remark here that the interaction of the reference system with the gravitational field of the planet is considered to be negligible by the authors of \cite{QRFIndefiniteMetric}, which is a necessary assumption to establish their result. While, as a consequence of that, the scenario they consider is not affected by the subtleties we raise, it is nevertheless instructive to analyze it in more detail, in order to illustrate our point.

Recall that the QRF formalism can be interpreted as implementing some symmetry on the level of observables (in the terminology of Subsection \hyperref[supp:sec:QRFsApplBosFerm]{I} in Supplementary Note 5 and the \hyperref[methods:sec:QRFstates]{Methods} section in the main text), while allowing arbitrary non-invariant density operators to describe quantum states. As the authors discuss, their argumentation applies in a physical regime where it makes sense to speak  of the positions of the involved physical systems and the associated global translations. The initial state of a reference system $R$, the mass $M$ and the test particle $S$ reads
\begin{align}
|\psi\rangle_{RMS}^{(R)}=|0\rangle_R\frac 1 {\sqrt{2}}\left(|x_M^{(1)}\rangle_M + |x_M^{(2)}\rangle_M\right)|x_S\rangle_S,
\label{eqQRFIndefinite1}
\end{align}
where $x_T$ denotes the position of subsystem $T\in\{R,M,S\}$ and $|\psi\rangle_T$ is the state of subsystem $T$. Now they use the QRF transformations described in Subsection \hyperref[supp:sec:QRFsApplBosFerm]{I} in Supplementary Note 5 above to move into a quantum coordinate system where the large mass $M$ is in a definite position:
\begin{align}
|\psi\rangle_{MRS}^{(M)}=|0\rangle_M\frac 1 {\sqrt{2}}\left(e^{-iPx_M^{(1)}}|-x_M^{(1)}\rangle_R |x_S-x_M^{(1)}\rangle_S + e^{-iPx_M^{(2)}}|-x_M^{(2)}\rangle_R |x_S-x_M^{(2)}\rangle_S\right),
\label{eqQRFIndefinite2}
\end{align}
where $P=0$ in~\cite{QRFIndefiniteMetric}. Note that the transition from~(\ref{eqQRFIndefinite1}) to~(\ref{eqQRFIndefinite2}) is achieved by applying the representation of the translation group~(\ref{eqRepOfTranslation}) to the initial state branch by branch. However, there are many such representations $V_P$, one for every total momentum $P$:
\begin{align}
V_P(T_d)|a\rangle|b\rangle|c\rangle=e^{iPd} |a+d\rangle|b+d\rangle|c+d\rangle.
\end{align}
In~(\ref{eqQRFIndefinite2}) above, we give the general result if the representation according to some value of $P$ is chosen.

From this for $P=0$, the authors of~\cite{QRFIndefiniteMetric} use standard methods to calculate the time-evolved state $|\psi(t)\rangle_{MRS}^{(M)}$, under the assumption that the reference system $R$ is far away enough such that it does not feel the gravitational pull of $M$. Let us drop this assumption here. We obtain
\begin{align}
|\psi(t)\rangle_{MRS}^{(M)}=|0\rangle_M\frac 1 {\sqrt{2}}\left(e^{-i\Phi^{(1)}} e^{-i P x_M^{(1)}}|-x_M^{(1)}+\delta x^{(1)}(t)\rangle_R |\tilde x_S^{(1)}(t)\rangle_S+e^{-i\Phi^{(2)}} e^{-i P x_M^{(2)}}|-x_M^{(2)}+\delta x^{(2)}(t)\rangle_R|\tilde x_S^{(2)}(t)\rangle_S\right),
\end{align}
where $\tilde x_S^{(i)}(t)=x_S(t)-x_M^{(i)}$, and $\delta x^{(i)}(t)$ denotes the position of the reference system $R$ at time $t$ in branch $i$, where the movement originates in the gravitational pull from $M$. We expect that $\delta x^{(i)}\approx 0$ if $R$ is very far away from $M$ in branch $i$, which is the case considered in~\cite{QRFIndefiniteMetric}. Then the authors of~\cite{QRFIndefiniteMetric} use the inverse QRF transformation to map back into $R$'s quantum coordinate system. In our more general case, doing so yields
\begin{align}
|\psi(t)\rangle_{RMS}^{(R)}=|0\rangle_R\frac 1 {\sqrt{2}}\left(e^{-iP\delta x^{(1)}(t)} e^{-i\Phi^{(1)}} |x_M^{(1)}-\delta x^{(1)}(t)\rangle_M|\tilde x_S^{(1)}(t)+x_M^{(1)}-\delta x^{(1)}(t)\rangle_S\right.\nonumber\\
\left.+e^{-i P \delta x^{(2)}(t)} e^{-i\Phi^{(2)}}|x_M^{(2)}-\delta x^{(2)}(t)\rangle_M|\tilde x_S(t)+x_M^{(2)}-\delta x^{(2)}(t)\rangle_S\right).
\label{eqMSEntangled}
\end{align}
Hence, the extended symmetry principle predicts entanglement between the test particle $S$ and the mass $M$, and a relative phase of the corresponding quantum state~(\ref{eqMSEntangled}) of
\begin{align}
\Delta\varphi:=\Phi^{(2)}-\Phi^{(1)}+P\left(\delta x^{(2)}(t)-\delta x^{(1)}(t)\right).
\end{align}
Unless $\delta x^{(1)}(t)=\delta x^{(2)}(t)$ (which holds for the scenario in~\cite{QRFIndefiniteMetric}, where both are zero), the predicted relative phase depends on $P$, which labels the choice of representation of the translation group. Consequently, the total momentum $P$ attains physical significance. Let us contrast this with a physical postulate formulated in~\cite{QRFIndefiniteMetric}: \textit{``Physical laws retain their form under quantum coordinate transformations.''} What the calculation above shows is that this cannot literally be true for \textit{all} quantum coordinate transformations, because some predictions depend on the chosen representation, which is labelled by $P$. Therefore, we suggest the following modification of the postulate:

\textbf{Physical laws retain their form under a suitable representation of the quantum coordinate transformation group.}

The extended covariance principle of~\cite{QRFIndefiniteMetric} is but one example of a larger class of quantum generalizations of symmetry principles considered in the literature, including also, for example, proposals for a quantum version of the equivalence principle~\cite{HardyQEP,GiacominiBrukner2020,GiacominiBrukner2022}, or a notion of quantum conformal symmetries~\cite{Kabel2024}. Our results underline the predictive power of such principles (we use them to rule out parastatistics), but they also motivate some caution: only \textit{some} quantum coordinate transformations will in general preserve the physical predictions.

This is demonstrated by our results on the permutation group. For $N$ Fermions, we have shown that all physical predictions (that is, the expectation values of antisymmetric observables) are invariant under the sign representation \begin{equation}
V_-(\vec{\pi})=\sum_{\lambda,j} {\rm sgn}(\pi_{\lambda,j})U_\lambda(\pi_{\lambda,j})\otimes P_{\lambda,j},
\label{supp:eq:QPermutation2}
\end{equation}
but not under the standard representation 
\begin{equation}
V_+(\vec{\pi})=\sum_{\lambda,j} U_\lambda(\pi_{\lambda,j})\otimes P_{\lambda,j},
\label{supp:eq:QPermutation1}
\end{equation}
of the quantum permutation group. Indeed, the latter type of quantum coordinate maps leads to experimentally detectable physical consequences, as described in the scheme by Roos et al.~\cite{Roos} for the detection of the Fermionic exchange phase. Hence,  additional physical arguments are in general needed to obtain the correct gauge phases~\cite{Nauenberg,MarlettoVedral,Dobkowski} for the construction of quantum coordinate transformations.

\section*{Supplementary Note 6 -- Relation to other notions of symmetry in previous work}

In this section, we describe how our notions of complete invariance and of invariance under quantum permutations (and their generalizations to other groups) are related to other notions of symmetry that have appeared in the literature: to strong/weak symmetry, gauge symmetries, and the notions of coherent vs.\ incoherent twirling for quantum reference frames.

\subsection{Strong versus weak and exact versus average symmetry}

The requirement that a density operator $\rho_S$ has full support on a subspace transforming under a one-dimensional representation of a compact group $G$, i.e.\ a subspace of the form $\cm_\lambda \otimes \cn_\lambda$ with $\cm_\lambda \simeq \Cl$, is sometimes known as \emph{strong symmetry} or \emph{exact symmetry}~\cite{Sala,Lessa_2024}. In this case every element $\ket{\psi_i}$ in an ensemble  $\{\ket{\psi_i}\}_i$ corresponding to the density operator $\rho_S = \sum_i p_i \ketbra{\psi_i}{\psi_i}_S$ is such that $\ket{\psi_i}$ transforms under the same one-dimensional irreducible  representation of $G$ for every element in the ensemble: $U_S(g) \ket{\psi_i}_S = e^{i \theta(g)} \ket{\psi_i}_S$ for all $\ket{\psi_i}_S \in \{\ket{\psi_i}_S\}_i$. This is equivalent to the requirement $U_S(g) \rho_S = e^{i \theta(g)} \rho_S$.

This notion is contrasted to the standard notion of invariance of a density operator, which just requires that the density operator is invariant: $U(g) \rho_S U^\dagger(g) = \rho_S$, and does not require individual elements in its decomposition to be invariant. Thus the density operator, interpreted as a proper mixture, is invariant as a statistical object, even though the individual elements of the ensemble need not be invariant. This standard invariance is sometimes known as \emph{weak symmetry} or \emph{average symmetry}, in contrast to the strong/exact symmetry described above.

In this language, Theorem \ref{main:th:CompleteInvariance} in the main text (generalized to arbitrary compact groups in~\Cref{TheoremCompleteInvariancegeneral}) establishes that complete invariance of $\rho_S$ is equivalent to exact symmetry of $\rho_S$.

The distinction between exact and average symmetry was first applied to Lindbladians in the literature on open quantum systems. Definition 4 of~\cite{Ma_average_2023}  states that for a strongly symmetric state $\rho_S = U_S(g) \rho_S$ there exists a purification $\ket{\psi}_{SA}$ which is invariant (as a ray) under $U_S(g) \otimes \I_A$. Our ~\Cref{TheoremCompleteInvariancegeneral} establishes the opposite direction of this statement also, showing that strong symmetry entails that every purification  (and indeed every extension) of the state is completely invariant.

The novel aspect of the present work, with regards to the notion of complete invariance, is its application to permutation invariance, thereby ruling out parastatistics. In the open systems literature it is typically applied to issues such as spontaneous symmetry breaking, typically involving $\Zl_2$ or $\U(1)$ symmetry.

\subsection{Complete invariance from local gauge symmetry}

At first glance, the requirement of complete invariance, namely invariance of an extension $\rho_{SA}$ under $\left(U_S(g) \otimes \I_A \right) \bullet \left(U_S^\dagger(g) \otimes \I_A\right)$, may seem too strong, or even arbitrary. Could it not be the case that the ancilla space $\ch_A$ carries a non-trivial representation $U_A(g)$, and could one not require invariance of $\rho_{SA}$ under $\left(U_S(g) \otimes U(g)_A \right)\bullet \left(U_S^\dagger(g) \otimes U_A^\dagger(g)\right) $ instead?

The latter type of invariance occurs in the case of spacetime symmetries, for example translation invariance, where physics is expected to remain invariant under a global translation of every system.  This is an instance of a global symmetry, which is typically understood to be distinct from a gauge symmetry. Global symmetries act as $\big(U_S(g) \otimes U_A(g) \big)\bullet \big(U_S^\dagger(g) \otimes U_A^\dagger(g)\big)$, and they map between physically equivalent descriptions of the quantum state under a lack of classical external reference frame for $G$ in the quantum information approach~\cite{Bartlett2007}. As can be seen immediately by tracing over $A$, invariance of a purification $\ket{\psi}_{SA}$ of $\rho_S$ under this symmetry implies that $\rho_S$ has weak symmetry under $U(g)$, i.e. $U_S(g) \rho_S U_S^\dagger(g) = \rho_S$. In~\Cref{thm:invariance_purification} (see also~\cite{Ma_average_2023}) we show the converse direction: for any invariant state $\rho_S$ one can find an ancilla $\ch_A$ and purification $\ket{\Psi}_{SA}$ such that $\ket{\Psi}_{SA}$ is invariant under $U_S(g) \otimes U_A(g)$. Thus, the requirement of weak symmetry of a state $\rho_S$ can be seen to follow from the existence of a purification which is invariant under a global symmetry.

A gauge symmetry is local if it acts non-trivially only on the system of interest, like e.g.\ the $\U(1)$ gauge symmetry of quantum electrodynamics. In the case of local gauge invariance, one expects invariance of the global quantum state under $\big(U_S(g) \otimes \I_A \big)\bullet \big(U_S^\dagger(g) \otimes \I_A\big)$. Permutation symmetry is a symmetry that acts locally on a system $S$ of $N$ particles.
Thus the requirement of complete invariance under permutations follows from the understanding of permutation invariance as a gauge symmetry and not a global symmetry.

\subsection{Global states in various approaches to quantum references frames}

There are a number of different frameworks for the study of quantum reference frames, including the perspective-neutral framework~\cite{HoehnUniverse2019,Vanrietvelde2020,Hamette2021,Vanrietvelde2023}, the operational framework~\cite{Miyadera2016,Loveridge2017,Loveridge2018,Loveridge2020,GlowackiThesis,Carette2023,Glowacki2024}, and the  quantum information-theoretic approach~\cite{Kitaev2004,Bartlett2007,GourSpekkens,GourMarvianSpekkens2009,MarvianThesis,MarvianSpekkens2014}, amongst others. In all these approaches there exists a representation $U(g)$ of the symmetry group $G$ of interest. In the perspective-neutral approach, the global state is assumed to transform under a trivial representation of $G$, i.e.\ $U(g) \ket \psi = \ket \psi$, whereas in the operational and the  quantum information-theoretic frameworks, the global state is assumed to be weakly symmetric under $U(g)$, i.e.\ $U(g) \rho U^\dagger(g) = \rho$. 

The requirement of  transforming trivially under the representation $U(g)$ is in fact equivalent to being strongly symmetric as we now show. Any state which is strongly symmetric and transforms under $U(g)$ as $e^{i \theta(g)}$ transforms trivially under the projectively equivalent representation $e^{ - i\theta(g)}U(g)$. Hence requiring strong symmetry under $U(g)$ is equivalent to requiring that the states transform trivially under some representation $V(g)$ which is projectively equivalent to $U(g)$.

Any non-invariant state $\rho$ can be mapped to a weakly symmetric state $\tilde \rho$  which gives the same probabilities on all invariant observables via the incoherent, or weak, twirl: $\tilde \rho = \int_{g \in G} U(g) \rho U^\dagger(g) dg$. A non-invariant state $\rho$ can be mapped to a strongly symmetric state $\rho_{\rm phys}$ which gives the same probabilities on all observables with full support on the trivial subspace via the coherent, or strong, twirl: $\rho_{\rm phys} = \int_{g \in G}\int_{g' \in G} U(g) \rho U^\dagger(g') dgdg'$. The action of the coherent twirl on pure states is $\ket \psi \mapsto \int_g U(g) \ket{\psi} dg$.

Given a physical scenario with a symmetry $G$, imposing the additional requirement of complete invariance (or not) allows one to single out the perspective-neutral framework (or not) as the relevant approach. As discussed above, this requirement could be motivated by an understanding of $U(g)$ as a gauge symmetry, or by appealing to the church of the larger Hilbert space. An explicit argument showing that the coherent twirl, but not the incoherent twirl, fulfills the requirements needed to define a gauge theory can be found in~\cite[Section II. B. 3]{HoehnKotechaMele}. In contrast, the incoherent twirl is defended in~\cite{Poulin_2005} by appeal to Bayesian probability.

The generalization of Theorems \ref{main:th:QuantPermInvRulesOutPara} and \ref{methods:th:MainTheorem2} in the main text beyond the quantum permutation group, showing how invariance under quantum permutations gives rise to the perspective-neutral subspace, will be given in upcoming work~\cite{Upcoming2}.

\end{document}